\documentclass[12pt]{article}

\usepackage{amsmath,amsthm,amssymb,amscd,epsfig,url}
\usepackage{fullpage}
\usepackage{times}
\newtheorem{theorem}{Theorem}[section]
\newtheorem{lemma}[theorem]{Lemma}
\newtheorem{proposition}[theorem]{Proposition}
\newtheorem{corollary}[theorem]{Corollary}
\newtheorem{conjecture}[theorem]{Conjecture}
\newtheorem{definition}[theorem]{Definition}
\newtheorem{example}[theorem]{Example}
\newtheorem{remark}[theorem]{Remark}

\def\<{{\langle}} \def\>{{\rangle}}       % Inner product
\title{On the Approximation Resistance of Balanced Linear Threshold Functions}
\author{
Aaron Potechin \\
University of Chicago
}
\date{\today}

\begin{document}
\maketitle

\begin{abstract}
In this paper, we show that there exists a balanced linear threshold function (LTF) which is unique games hard to approximate, refuting a conjecture of Austrin, Benabbas, and Magen. We also show that the almost monarchy predicate on $k$ variables is approximable for sufficiently large $k$.
\end{abstract}

\thispagestyle{empty}
.\newpage

\section{Introduction}
Constraint satisfaction problems (CSPs) are a central topic of study in computer science. A fundamental question about CSPs is as follows. Given a CSP where each constraint has the form of some predicate $P$ and almost all of the constraints can be satisfied, is there a randomized polynomial time algorithm which is guaranteed to do significantly better in expectation than a random assignment? If so, then we say that the predicate $P$ is approximable. If not, then we say that $P$ is approximation resistant.

There is a large body of work on this question. On the algorithmic side, Goemans and Williamson's \cite{GW95} breakthrough algorithm for MAX CUT using semidefinite programming implies that all predicates $P$ on two boolean variables are approximable. H{\aa}stad \cite{Has08} later generalized this result, showing that all 2-CSPs are approximable. For larger arity predicates, many individual predicates have been shown to be approximable, see e.g. \cite{Hast05,Has07,ABM10}, but we have few general criteria for showing that a predicate is approximable. 

On the hardness side, in a breakthrough work H{\aa}stad \cite{Has01} used his 3-bit PCP theorem to prove that 3-SAT and 3-XOR are NP-hard to approximate. In another breakthrough work, Khot \cite{Khot02} discovered that if we assume the unique games problem is hard then we can show further inapproximability results. Building on this work, Khot, Kindler, Mossel, and O'Donnell \cite{KKMO07} showed that if we assume the unique games problem is hard then the Goemans-Williamson algorithm for MAX CUT is optimal and Austrin and Mossel \cite{AM09} showed that any predicate which has a balanced pairwise independent distribution of solutions is unique games hard to approximate. Chan \cite{Chan16} later strengthened this to NP-hardness under the stronger condition that there exists a pairwise independent subgroup. 

In 2008, Raghavendra \cite{Rag08} proved a dichotomy theorem for the hardness of CSPs. Either a standard semidefinite program (SDP) gives a better approximation ratio than a random assignment or it is unique games hard to do so. However, for any given CSP it may be extremely hard to decide which case holds. In fact, it is not even known whether it is decidable!

In this paper, we investigate the approximability of balanced linear threshold functions (LTFs), a simple class of predicates for which little was previously known. Prior to this work, it was an open problem whether there is any balanced LTF which is unique games hard to approximate (Austrin, Benabbas, and Magen \cite{ABM10} conjectured that there are none) and the only balanced LTFs which were known to be approximable were the monarchy predicate \cite{ABM10} and LTFs which are close to the majority function \cite{Hast05}. 
\subsection{Our results}
Our main result is the following theorem which refutes the conjecture of Austrin, Benabbas, and Magen \cite{ABM10}.
\begin{definition}
A balanced linear threshold function (LTF) is an LTF with no constant term i.e. a function of the form $f(x_1,\dots,x_k) = sign(\sum_{i=1}^{k}{{c_i}{x_i}})$.
\end{definition}
\begin{theorem}
There exists a predicate $P$ which is a balanced linear threshold function and is unique games hard to approximate.
\end{theorem}
\begin{remark}
This predicate $P$ is very different from other predicates which are known to be unique games hard to approximate. To the best of our knowledge, for all other predicates $P'$ which have previously been shown to be unique games hard to approximate, either $P'$ has a balanced pairwise independent distribution of solutions or (like the predicate studied by Guruswami, Lewin, Sudan, and Trevisan \cite{GLST98}) can be easily reduced to such a predicate. As a consequence, while linear degree sum of squares lower bounds are known for all other predicates $P'$ which have previously been shown to be unique games hard to approximate, we do not have sum of squares lower bounds for approximating this predicate $P$.
\end{remark}
We also prove the following approximability result.
\begin{theorem}
The almost monarchy predicate on $k$ variables is approximable if $k$ is sufficiently large.
\end{theorem}
\subsection{Outline}
The remainder of the paper is organized as follows. In Section \ref{preliminariessection} we give some preliminary definitions. In section \ref{approximationresistancecriteriasection} we recall the known criteria for proving that a predicate $P$ is unique games hard to approximate. In section \ref{mainconstructionsection}, we construct a balanced LTF which is unique games hard to approximate, proving our main result.

After proving our main result, we switch gears and analyze balanced LTFs which are approximable. In section \ref{generalapproximationsection} we give a description based on the analysis of Khot, Tulsiani, and Worah \cite{KTW13} for the space of approximation algorithms we should search over. In section \ref{monarchysection} we use this description to give a simpler approximation algorithm for the monarchy predicate. Finally, in section \ref{almostmonarchysection} we give an approximation algorithm for the almost monarchy predicate on $k$ variables when $k$ is sufficiently large.
\section{Preliminaries}\label{preliminariessection}
In this section, we give some preliminary definitions. In particular, we define what we mean by predicates, constraints, and approximation resistance and recall the unique games problem and the unique games conjecture.
\begin{definition}
A boolean predicate $P$ of arity $k$ is a function $P: \{-1,1\}^{k} \to \{-1,1\}$
\end{definition}
\begin{definition}
A boolean constraint $C$ is a function $C:\{-1,+1\}^{n} \to \{-1,+1\}$. We say a boolean constraint $C$ is satisfied by an input $x \in \{-1,+1\}^n$ if $C(x) = 1$.
\end{definition}
\begin{definition}
Let $P$ be a boolean predicate. We say that a boolean constraint $C: \{-1,+1\}^{n} \to \{-1,+1\}$ has form $P$ if there exists a map $\phi: [1,k] \to [1,n]$ and signs $(z_1,\dots,z_k) \in \{-1,+1\}^k$ such that 
\[
C(x_1,\dots,x_n) = P(z_1x_{\phi(1)},\dots,z_{k}x_{\phi(k)})
\]
\end{definition}
\begin{definition}
We say that a boolean predicate $P$ is approximable if there exists an $\epsilon > 0$ such that there is a polynomial time algorithm which takes a CSP with $m$ constraints of form $P$ as input and can distinguish between the following two cases:
\begin{enumerate}
\item At least $(1-\epsilon)m$ of the constraints can be satisfied.
\item At most $(r_P + \epsilon)m$ of the constraints can be satisfied where $r_P = \mathbb{E}_{x \in \{-1,+1\}^k}\left[\frac{P(x) + 1}{2}\right]$ is the probability that a random assignment satisfies $P$.
\end{enumerate}
If neither case holds then the output of the algorithm can be arbitrary. If no such algorithm exists for any $\epsilon > 0$ then we say that $P$ is approximation resistant.
\end{definition}
Currently, we can only show that predicates are approximation resistant under an assumption such as $P \neq NP$ or the unique games conjecture.
\begin{definition}
We say that a boolean predicate $P$ is NP-hard to approximate if for all $\epsilon > 0$ it is NP-hard to take a a CSP with $m$ constraints of form $P$ as input and distinguish between the following two cases:
\begin{enumerate}
\item At least $(1-\epsilon)m$ of the constraints can be satisfied.
\item At most $(r_P + \epsilon)m$ of the constraints can be satisfied where $r_P = \mathbb{E}_{x \in \{-1,+1\}^k}\left[\frac{P(x) + 1}{2}\right]$ is the probability that a random assignment satisfies $P$.
\end{enumerate}
\end{definition}
\begin{definition}
In the unique games problem with $t$ labels, we are given a graph $G$ together with a bijective map $\phi_e:[1,t] \to [1,t]$ for each edge $e \in E(G)$. We are then asked to assign a label $x_v \in [1,t]$ to each vertex $v$ and maximize $|e = (v,w) \in E(G): \phi_e(x_v) = x_w|$ (i.e. the number of edge constraints which are satisfied)
\end{definition}
Khot \cite{Khot02} made the following conjecture, known as the unique games conjecture
\begin{conjecture}[Unique Games Conjecture]
For all $\epsilon > 0$ there exists a $t$ such that it is NP-hard to take a unique games problem with $t$ labels and distinguish between the following cases:
\begin{enumerate}
\item $max_{\{x_v\}}\{|e = (v,w) \in E(G): \phi_e(x_v) = x_w|\} \geq (1 - \epsilon)|E(G)|$. In other words, at least $(1 - \epsilon)|E(G)|$ of the edge constraints can be satisfied.
\item $max_{\{x_v\}}\{|e = (v,w) \in E(G): \phi_e(x_v) = x_w|\} \leq \epsilon|E(G)|$. In other words, at most $\epsilon|E(G)|$ of the edge constraints can be satisfied.
\end{enumerate}
\end{conjecture}
\begin{definition}
We say that a predicate $P$ is unique games hard to approximate if for any $\epsilon_1 > 0$ and integer $t$ there exists an $\epsilon_2 > 0$ such that there is a polynomial time reduction from the problem of taking a unique games instance with $t$ labels and distinguishing between the following two cases:
\begin{enumerate}
\item At least $(1 - \epsilon_1)|E(G)|$ edge constraints can be satisfied.
\item At most ${\epsilon_1}|E(G)|$ edge constraints can be satisfied.
\end{enumerate}
to the problem of taking a CSP on $m$ constraints where the predicates have form $P$ and distinguishing between the following two cases:
\begin{enumerate}
\item At least $(1-\epsilon_2)m$ of the constraints can be satisfied.
\item At most $(r_P + \epsilon_2)m$ of the constraints can be satisfied where $r_P = \mathbb{E}_{x \in \{-1,+1\}^k}\left[\frac{P(x) + 1}{2}\right]$ is the probability that a random assignment satisfies $P$.
\end{enumerate}
\end{definition}
\section{Criteria for Approximation Resistance}\label{approximationresistancecriteriasection}
In this section, we recall known criteria for proving that a predicate $P$ is unique games hard to approximate and define perfect integrality gap instances, which is a special case of Raghavenra's criterion as is the criterion we will use.
\subsection{The standard SDP for CSPs}
We first recall the standard SDP described by Raghavendra \cite{Rag08} which gives an approximation ratio better than the random assignment whenever it is not unique games hard to do so.

The idea behind the SDP is as follows. The SDP searches for a set of global biases and pairwise biases $B = \{b_i: i \in [1,n]\} \cup \{b_{ij}: i,j \in [1,n], i < j\}$ for the variables where each possible $B$ is evaluated as follows. For each constraint $C$, we find a distribution of assignments to the variables involved in $C$ which matches $B$ and maximizes the probability that $C$ is satisfied. This probability is the score of $B$ on the constraint $C$. The goal of the SDP is to find a $B$ which maximizes the sum of the scores of $B$ on all of the constraints.

To make this rigorous, we make the following definitions.
\begin{definition}
Given $x = (x_1,\dots,x_k) \in \{-1,+1\}^{k}$, we define the point $p_x \in \{-1,+1\}^{k + \binom{k}{2}}$ to be 
\[
p_x = (x_{i}: i \in [1,k]) \text{ concatenated with } (x_{i_1}x_{i_2}: i_1,i_2 \in [1,k], i_1 < i_2)
\] 
where the pairs $i_1,i_2$ are in lexicographical order.
\end{definition}
\begin{definition}
Given the biases and pairwise biases $B = \{b_i\} \cup \{b_{ij}\}$ and a set of indices $I = \{i_1,\dots,i_k\}$ in increasing order, we define the point 
$p_{B,I} \in \mathbb{R}^{k + \binom{k}{2}}$ to be 
\[
p_{B,I} = (b_{i_j}: j \in [1,k]) \text{ concatenated with } (b_{i_{j_1}i_{j_2}}: j_1,j_2 \in [1,k], j_1 < j_2)
\]
where the pairs $j_1,j_2$ are in lexicographical order.
\end{definition}
\begin{definition}
Given a set of points $X = \{p_1,\dots,p_m\}$, we define 
\[H(X) = \{\sum_{i=1}^{m}{{a_i}p_i}: \forall i \in [1,m], a_i \in [0,1], \sum_{i=1}^{m}{a_i} = 1\}
\] 
to be the convex hull of the points in $X$.
\end{definition}
\begin{definition}
Given a constraint $C$ on $k$ variables $\{x_{i_1},\dots,x_{i_k}\}$, we define the KTW polytope ${KTW}_C$ of $C$ to be 
\[
{KTW}_C = H(\{p_x: x = (x_{i_1},\dots,x_{i_k}) \in \{-1,+1\}^{k}, C(x) = 1\})
\]
Also, we define the polytope ${ALL}_k$ to be 
\[
{ALL}_k = H(\{p_x: x = (x_{i_1},\dots,x_{i_k}) \in \{-1,+1\}^{k}\})
\]
\end{definition}
\begin{remark}
We call this polytope the KTW polytope because Khot, Tulsiani, and Worah \cite{KTW13} highlighted the central role this polytope plays in determining whether a predicate $P$ is strongly approximation resistant or not. That said, it should be noted that similar polytopes appeared in previous papers (see e.g. \cite{AH11,AH13,AK13}).
\end{remark}
\begin{definition}
We now define the standard SDP for a set of constraints $\{C_1,\dots,C_m\}$ on variables $\{x_1,\dots,x_n\}$. We have the following variables:
\begin{enumerate}
\item We have variables $\{b_{i}: i \in [1,n]\}$ and $\{b_{ij}: i,j \in [1,n], i < j\}$. We take the matrix $B$ so that 
\begin{enumerate}
\item $B_{00} = 1$
\item $\forall i \in [1,n], B_{0i} = B_{i0} = b_i$
\item $\forall i \in [1,n], B_{ii} = 1$
\item For all $i,j \in [1,n]$ such that $i < j$, $B_{ij} = B_{ji} = b_{ij}$
\end{enumerate}
\item For each constraint $C_i$, we have variables $a_{{C_i}1},a_{{C_i}2},p_{{C_i}1},p_{{C_i}2}$
\end{enumerate}
With these variables, we are trying to maximize $\sum_{i=1}^{m}{a_{{C_i}1}}$ subject to the following constraints:
\begin{enumerate}
\item $B \succeq 0$
\item For each $i \in [1,m]$, 
\begin{enumerate}
\item $a_{{C_i}1},a_{{C_i}2} \in [0,1]$, $a_{{C_i}1} + a_{{C_i}2} = 1$
\item $p_{{C_i}1} \in {KTW}_{C_i}$, $p_{{C_i}2} \in {ALL}_{k_i}$ where $k_i$ is the arity of the constraint $C_i$
\item $p_{B,I_i} = a_{{C_i}1}p_{{C_i}1} + a_{{C_i}2}p_{{C_i}2}$ where $I_i = \{i_j: j \in [1,k_i]\}$ is the set of indices which $C_i$ depends on.
\end{enumerate}
\end{enumerate}
\end{definition}
Raghavendra \cite{Rag08} proved the following theorem:
\begin{theorem}
A predicate $P$ is unique games hard to approximate if and only if for all $\epsilon > 0$ there is a CSP instance with $m$ constraints of the form $P$ such that 
\begin{enumerate}
\item The standard SDP gives a value of at least $(1 - \epsilon)m$.
\item At most $(r_P + \epsilon)m$ of the constraints can be satisfied where $r_P$ is the probability that a random assignment satisfies the predicate $P$.
\end{enumerate}
\end{theorem}
\subsection{The KTW criterion}
Khot, Tulsiani, and Worah \cite{KTW13} found an alternative criterion for showing that a predicate $P$ is unique games hard to approximate. That said, the KTW criterion doesn't quite match Raghavendra's criterion. Rather, the KTW criterion corresponds to the slightly stronger statement that $P$ is unique games hard to weakly approximate.
\begin{definition}
We say that a boolean predicate $P$ is weakly approximable if there exists an $\epsilon > 0$ such that there is a polynomial time algorithm which takes a CSP with $m$ constraints of form $P$ as input and can distinguish between the following two cases:
\begin{enumerate}
\item At least $(1-\epsilon)m$ of the constraints can be satisfied.
\item Every $x$ satisfies between $(r_P - \epsilon)m$ and $(r_P + \epsilon)m$ of the constraints where $r_P = \mathbb{E}_{x \in \{-1,+1\}^k}\left[\frac{P(x) + 1}{2}\right]$ is the probability that a random assignment satisfies $P$. In other words, no assignment does much better or worse than random.
\end{enumerate}
If neither case holds then the output of the algorithm can be arbitrary. If no such algorithm exists for any $\epsilon > 0$ then we say that $P$ is strongly approximation resistant.
\end{definition}
\begin{definition}
We say that a predicate $P$ is unique games hard to weakly approximate if for any $\epsilon_1 > 0$ and integer $t$ there exists an $\epsilon_2 > 0$ such that there is a polynomial time reduction from the problem of taking a unique games instance with $t$ labels and distinguishing between the following two cases:
\begin{enumerate}
\item At least $(1 - \epsilon_1)|E(G)|$ edge constraints can be satisfied.
\item At most ${\epsilon_1}|E(G)|$ edge constraints can be satisfied.
\end{enumerate}
to the problem of taking a CSP on $m$ constraints where the predicates have form $P$ and distinguishing between the following two cases:
\begin{enumerate}
\item At least $(1-\epsilon_2)m$ of the constraints can be satisfied.
\item All assignments satisfy between $(r_P - \epsilon_2)m$ and $(r_P + \epsilon_2)m$ contraints where $r_P = \mathbb{E}_{x \in \{-1,+1\}^k}\left[\frac{P(x) + 1}{2}\right]$ is the probability that a random assignment satisfies $P$.
\end{enumerate}
\end{definition}
To write down the KTW criterion, we need the following definitions from \cite{KTW13}:
\begin{definition}
For a measure $\Lambda$ on $KTW_P$ and a subset $S \subseteq [1,k]$, let $\Lambda_S$ denote the projection of $\Lambda$ onto the coordinates of $S$. For a permutation $\pi: S \to S$ and a choice of signs $z \in \{-1,1\}^{|S|}$, let $\Lambda_{S,\pi,z}$ denote the measure $\Lambda_S$ after permuting the the indices in $S$ according to $\pi$ and then (possibly) negating the coordinates according to multiplication by $\{z_i\}_{i \in S}$
\end{definition}
\begin{definition}[Definition 1.1 of KTW \cite{KTW13}]
Let $\mathcal{A}_s$ be the family of all predicates (of all arities) $: \{-1,1\}^k \to \{0,1\}$ such that there is a probability measure $\Lambda$ on $KTW_P$ such that for every $1 \leq t \leq k$, the signed measure
\[
\Lambda_{P}^{(t)} := \mathbb{E}_{S:|S| = t}\mathbb{E}_{\pi:[1,t] \to [1,t]}\mathbb{E}_{(z_1,\dots,z_t) \in \{-1,1\}^t}\left[\left(\prod_{i=1}^{t}{z_i}\right) \cdot \hat{P}_S \cdot \Lambda_{S,\pi,z}\right]
\]
vanishes identically. If so, $\Lambda$ itself is said to vanish.
\end{definition}
Khot, Tulsiani, and Worah \cite{KTW13} prove the following theorem
\begin{theorem}
A predicate $P$ is unique games hard to weakly approximate if and only if $P \in \mathcal{A}_s$
\end{theorem}
\begin{remark}
The KTW criterion has the advantage that it can be expressed only in terms of the polytope $KTW_P$ and the Fourier characters of $P$, but like Raghavendra's criterion, it is unknown whether it is even decidable.
\end{remark}
\subsection{Unique Games Hardness of Approximation from a Balanced Pairwise Independent Distribution of Solutions}
We now observe (as has been observed before) that if a predicate has a balanced pairwise indepedent distribution of solutions, which was shown to imply unique games hardness by Austrin and Mossel \cite{AM09}, then it satisfies both of Raghavendra's criterion and the KTW criterion.
\begin{definition}
Let $P: \{-1,1\}^{k} \to \{-1,+1\}$ be a boolean predicate. We say that a distribution $D$ is a balanced pairwise independent distribution of solutions for $P$ if 
\begin{enumerate}
\item $D$ is supported on $\{x \in \{-1,1\}^{k}: P(x) = 1\}$
\item For all $i \in [1,k]$, $E_{D}[x_i] = 0$
\item For all distinct $i,j \in [1,k]$, $E_{D}[{x_i}{x_j}] = 0$
\end{enumerate}
\end{definition}
\begin{example}
Consider the 3-XOR predicate $P(x_1,x_2,x_3) = {x_1}{x_2}{x_3}$. The uniform distribution on the four solutions $(1,1,1), (1,-1,-1), (-1,1,-1), (-1,-1,1)$ is a balanced pairwise independent distribution of solutions for $P$.
\end{example}
\begin{lemma}
If a predicate $P$ has a balanced pairwise independent distribution of solutions then it satisfies both Raghavendra's criterion for being unique games hard to approximate and the KTW criterion for being unique games hard to weakly approximate.
\end{lemma}
\begin{proof}
To see that $P$ satisfies Raghavendra's criterion for being unique games hard to approximate, consider the instance with the $2^k$ constraints $\{P({b_1}{x_1},\dots,{b_k}{x_k}) = 1: (b_1,\dots,b_k) \in \{-1,1\}^k\}$ and observe that 
\begin{enumerate}
\item Every assignment satisfies exactly $2^k{r_P}$ constraints where $r_P$ is the probability that a random assignment satisfies $P$.
\item The standard SDP has value $2^k$ and this can be achieved by setting $b_i = 0$ for all $i$ and $b_{ij} = 0$ for all $i < j$. 
\end{enumerate}
To see that $P$ satisfies the KTW criterion for being unique games hard to weakly approximate, note that having a pairwise independent distribution of solutions implies that $\vec{0} \in KTW_P$ so we can take $\Lambda$ to be the probability measure which is $\vec{0}$ with probability $1$. We now have that 
\[
\Lambda_{P}^{(t)} = \mathbb{E}_{S:|S| = t}\mathbb{E}_{\pi:[1,t] \to [1,t]}\mathbb{E}_{(z_1,\dots,z_t) \in \{-1,1\}^t}\left[\left(\prod_{i=1}^{t}{z_i}\right) \cdot \hat{P}_S \cdot \Lambda_{S,\pi,z}\right] = 0
\]
because $\Lambda_{S,\pi,z}$ will always be the measure which is $\vec{0}$ with probability $1$.
\end{proof}
\subsection{Perfect Integrality Gap Instances}
While the criterion of having a balanced pairwise independent distribution of solutions is simple and powerful, it does not capture all predicates $P$ which are unique games hard to approximate and we need a stronger criterion. The criterion we use is that there is a perfect integrality gap instance for the standard SDP which consists of functions of the form $P$.
\begin{definition}
Let $\{f_a: a \in [1,m]\}$ be a multi-set of $\pm{1}$-valued functions on $\{-1,+1\}^n$ where each $f_a$ depends on a subset $V_a \subseteq \{x_1,\dots,x_n\}$ of the variables. We say that $\{f_a: a \in [1,m]\}$ is a perfect integrality gap instance for the standard SDP if 
\begin{enumerate}
\item There exist biases and pairwise biases $B = \{b_i, i \in [1,n]\} \cup \{b_{ij}: i,j \in [1,n], i < j\}$ such that for all $a$, there exists a distribution $D_{a}$ on satisfying assigments to $f_a$ which matches $B$, i.e. 
\begin{enumerate}
\item $\forall i \in V_a, E_{D_{a}}[x_i] = b_i$
\item For all $i,j \in V_i$ such that $i < j$, $E_{D_{a}}[x_{i}x_{j}] = b_{ij}$
\end{enumerate}
\item $\sum_{a=1}^{m}{f_a(x_1,\dots,x_n)} = cm$ for some constant $c \in [-1,1]$
\end{enumerate}
For brevity, we will just say ``perfect integrality gap instance'' and omit writing ``for the standard SDP''
\end{definition}
\begin{remark}
In this paper, it is sufficient to consider perfect integrality gap instances where all of the functions $f_a$ are constraints on the same set of variables $x_1,\dots,x_k$. However, we do not require this in the definition. 
\end{remark}
\begin{example}
As shown in the previous subsection, if $P: \{-1,+1\}^{k} \to \{-1,+1\}$ is a predicate which has a balanced pairwise independent distribtion of solutions then $\{P({b_1}{x_1},\dots,{b_k}{x_k}): (b_1,\dots,b_k) \in \{-1,1\}^k\}$ is a perfect integrality gap instance.
\end{example}
\begin{example}
The predicate $P(x_1,x_2,x_3,x_4) = -\frac{1-x_1}{2}{x_2}{x_3} - \frac{1+x_1}{2}{x_2}{x_4}$ was shown to be NP-hard to approximate by Guruswami, Lewin, Sudan, and Trvisan \cite{GLST98}. This predicate does not have a pairwise independent distribution of solutions but $\{P(x_1,x_2,x_3,x_4), P(x_1,-x_2,x_3,x_4)\}$ is a perfect integrality gap instance because
\begin{enumerate}
\item For all $(x_1,x_2,x_3,x_4) \in \{-1,+1\}^4$, exactly one of $P(x_1,x_2,x_3,x_4)$ and $P(x_1,-x_2,x_3,x_4)$ will be $1$.
\item Taking the biases and pairwise biases $B = \{b_i: i \in [1,4]\} \cup \{b_{ij}: i,j \in [1,4], i < j\}$ to all be zero except for $b_{34} = -1$, if we take $D_1$ to be the uniform distribution on $\{(1,1,1,-1),(1,-1,-1,1),(-1,1,-1,1),(-1,-1,1,-1)\}$ and take $D_2$ to be the uniform distribution on $\{(-1,1,1,-1),(1,-1,1,-1),(1,1,-1,1),(-1,-1,-1,1)\}$ then $D_1$ is a distribution of solutions to $P(x_1,x_2,x_3,x_4) = 1$ which matches $B$ and $D_2$ is a distribution of solutions to $P(x_1,-x_2,x_3,x_4) = 1$ which matches $B$.
\end{enumerate}
\end{example}
\begin{remark}
While the GLST predicate does not have a balanced pairwise indepedent distribution of solutions, if we set $x_4 = -x_3$ then it becomes the 3-XOR predicate, so its hardness essentially comes from the hardness of 3-XOR.
\end{remark}
\begin{lemma}
If there is a perfect integrality gap insatnce $\{f_a: a \in [1,m]\}$ of functions of form $P$ then $P$ satisfies both Raghavendra's criterion for being unique games hard to approximate and the KTW criterion for being unique games hard to weakly approximate.
\end{lemma}
\begin{proof}
If there is a perfect integrality gap instance $\{f_a: a \in [1,m]\}$ of functions of form $P$ then Raghavendra's criterion for being unique games hard to approximate is satisfied by definition. In particular, 
\begin{enumerate}
\item Every assignment satisfies exactly $m{r_P}$ constraints where $r_P$ is the probability that a random assignment satisfies $P$.
\item The standard SDP has value $m$ and this can be achieved with the given biases $\{b_i\}$ and pairwise biases $\{b_{ij}\}$
\end{enumerate}
The proof that the KTW criterion is satisfied requires carefully putting the definitions together. Since this is somewhat technical, we defer this proof to Appendix \ref{verifyingKTWappendix}.
\end{proof}
\section{An approximation resistant LTF}\label{mainconstructionsection}
In this section, we construct a predicate $P$ which is a balanced linear threshold function and is unique games hard to approximate.
\subsection{Overview of the construction}
\subsubsection{High level overview}
The high level idea for our construction is as follows
\begin{enumerate}
\item The core of our construction is a predicate $P = sign(l)$ (where $l$ is a linear form) together with a set of constraints of the form $P$ which gives a perfect integrality gap for the standard SDP (generalized to bounded integer valued variables) with the following adjustments:
\begin{enumerate}
\item Rather than considering all possible solutions, we only consider solutions from some set $V$.
\item The variables may take integer values rather than just values in $\{-1,+1\}$.
\end{enumerate}
\item We handle the first adjustment by adding constraints to our linear form $l$ which are only satisfed by our subset $V$ of possible solutions. However, for technical reasons this gives us two predicates $P_1 = sign(l_1)$ and $P_2 = sign(l_2)$ instead of just one predicate.
\item We handle the second adjustment by encoding our variables in unary.
\item From the two predicates $P_1 = sign(l_1)$ and $P_2 = sign(l_2)$, we obtain a single predicate $P_3 = sign(l_3)$ which ``simulates'' both $P_1$ and $P_2$. This predicate $P_3$ is our balanced LTF which is unique games hard to approximate.
\end{enumerate}
\subsubsection{Technical overview}
More precisely, we will construct an approximation resistant balanced LTF as follows:
\begin{enumerate}
\item In subsection \ref{coresubsection}, we find a set of linear forms $\{l_a\}$ together with a set of possible solution vectors $V$ in $\mathbb{Z}^k$ (but not necessarily in $\{-1,+1\}^k$) such that 
\begin{enumerate}
\item There exist values $\{c_{i}\}$, $\{c_{ii}\}$, $\{c_{ij}\}$ such that for all $a$, there is a distribution $D_a$ over the vectors $\{v: v \in V, l_{a}(v) > 0\}$ such that 
$\forall i, E_{D_a}[v_i] = c_i$, $\forall i, E_{D_a}[v^2_i] = c_{ii}$, and $\forall i < j, E_{D_a}[{v_i}{v_j}] = c_{ij}$
\item For every vector $v \in V$, $l_a(v) > 0$ for exactly half of the $\{l_a\}$.
\end{enumerate}
In other words, $\{sign(l_a)\}$ would be a perfect integrality gap instance except for the following issues:
\begin{enumerate}
\item Instead of considering all possible inputs, we only consider a subset $V$ of the possible inputs.
\item The input vectors in $V$ have entries in $\mathbb{Z}$ rather than $\{-1,+1\}$
\end{enumerate}
All of these linear forms $l_a$ will be the same as some linear form $l$ up to permuting the variables.
\item In subsections \ref{constraintsubsection} and \ref{specificationsubsection}, we show how to obtain new linear forms $\{l'_{a1}\} \cup \{l'_{a2}\}$ and a set solution vectors $V'$ such that 
\begin{enumerate}
\item If $v' \notin V'$ then for all $a$, $l'_{a1}(v')l'_{a2}(v') < 0$.
\item For all $a,\sigma$ and all $v_{b} \in V$ there is a vector $v'_{b\sigma} \in V'$ such that $l'_{a1}(v'_{b\sigma}) = l'_{a2}(v'_{b\sigma}) = l_a(v_b)$
\item There exist values $\{c'_{i}\}$, $\{c'_{ii}\}$, $\{c'_{ij}\}$ such that for all $a$, there is a distribution $D'_a$ over the vectors $\{v': v' \in \{v'_{b\sigma}\}, l'_{a1}(v') > 0, l'_{a2}(v') > 0\}$ such that 
$\forall i, E[v'_i] = c'_i$, $\forall i, E[v'^2_i] = c'_{ii}$, and $\forall i < j, E[{v'_i}{v'_j}] = c'_{ij}$
\end{enumerate}
In other words, the linear forms $\{l'_{a1}\} \cup \{l'_{a2}\}$ will test whether $v'$ is in the set of specified solution vectors $V'$. If not, then $v$ will automatically make exactly half of the linear forms $\{l'_{a1}\} \cup \{l'_{a2}\}$ positive. If $v'$ is in the set of possible solution vectors then the linear forms $\{l'_{a1}\} \cup \{l'_{a2}\}$ will behave like the linear forms $\{l_a\}$. Since each solution vector $v \in V$ makes exactly half of the linear forms $\{l_a\}$ positive, this implies that every $v'$ makes exactly half of the linear forms $\{l'_{a1}\} \cup \{l'_{a2}\}$ positive.

Thus, taking the linear forms $\{l'_{a1}\} \cup \{l'_{a2}\}$ fixes the issue of considering only a subset of the possible inputs. However, there is now a new issue. The linear forms $\{l'_{a1}\}$ will be the same as some linear form $l'_1$ up to permuting the variables and the linear forms $\{l'_{a2}\}$ will be the same as some linear form $l'_2$ up to permuting the variables but we will not have $l'_1 = l'_2$. Thus, we will have two different LTFs.
\item In subsection \ref{unarysubsection}, we describe how we can make our vectors have $\pm{1}$ values by expressing all of the coordinates in unary. We also make our linear threshold functions balanced by replacing $1$ with a special variable $x_{one}$ which we always expect to be $1$.
\item Finally, in subsections \ref{mergingsubsection} and \ref{balancingsubsection} we describe how to take two different balanced linear forms $l_1$ and $l_2$ on $\{x_1,\dots,x_k\}$ where $k$ is a power of $2$ (we need this condition for technical reasons) and construct a third balanced linear form $l_3$ on $\{x_1,\dots,x_{k^2}\}$ such that given a perfect integrality gap instance $\{f_a\}$ where each $f_a$ is on the variables $\{x_1,\dots,x_k\}$ and has the form $sign(l_1)$ or $sign(l_2)$, we can construct a perfect integrality gap instance $\{f'_b\}$ where each $f'_b$ is on the variables $\{x_1,\dots,x_{k^2}\}$ and has the form $sign(l_3)$
\end{enumerate}
Putting everything together, the predicate $sign(l_3)$ is unique games hard to weakly approximate.
\subsection{Core of the construction}\label{coresubsection}
As described in the overview, the core of the construction is a predicate $P = sign(l)$ (where $l$ is a linear form) together with a set of constraints of form $P$ which gives a perfect integrality gap instance for the standard SDP (generalized to bounded integer valued variables) with the following adjustments:
\begin{enumerate}
\item We restrict our attention to a subset $V$ of the possible solution vectors.
\item The variables may take integer values rather than just values in $\{-1,+1\}$.
\end{enumerate}
More precisely, we have the following lemma:
\begin{lemma}\label{corelemma}
There exists a set of linear forms $\{l_a\}$, a set of possible solution vectors $V$, and values $\{c_i\},\{c_{ii}\},\{c_{ij}\}$ such that
\begin{enumerate}
\item $\forall v \in V, Pr_{l \in \{l_a\}}[sign(l(v)) > 0] = \frac{1}{2}$
\item For all $a$ there exists a distribution $D_a$ such that
\begin{enumerate}
\item $D_a$ is supported on the set $\{v: v \in V, l_a(v) > 0\}$
\item $\forall i, E_{v \in D_a}[v_i] = c_i$
\item $\forall i, E_{v \in D_a}[v^2_i] = c_{ii}$
\item $\forall i,j, E_{v \in D_a}[{v_i}{v_j}] = c_{ij}$
\end{enumerate}
\end{enumerate}
In fact, we may take $V$ and the values $\{c_i\},\{c_{ii}\},\{c_{ij}\}$ to be symmetric under permutations of the input variables and have that all of the $\{l_a\}$ are the same as some linear form $l$ up to permuting the input variables.
\end{lemma}
\begin{proof}
We can take the following linear forms $\{l_a\}$, solution vectors $V$, and values $\{c_i\},\{c_{ii}\},\{c_{ij}\}$
\begin{enumerate}
\item $\{l_a\} = \{x_1 + 1.5 - 1.6(x_2 + x_3)/299, x_2 + 1.5 -1.6(x_3 + x_4)/299, x_3 + 1.5 -1.6(x_1 + x_4)/299, x_4 + 1.5 -1.6(x_1 + x_2)/299\}$
\item $V$ is the set of all permutations of the following vectors: 
\[\{(299,0,0,0), (-1,-1,-7,-7), (64,-1,-2,-2)\}\]
\item $\forall i, c_i = 0$, $\forall i, c_{ii} = \frac{1345500}{4500} = 299$, and $\forall i \neq j, c_{ij} = 0$
\end{enumerate}
We first check condition 1. 

For the vector $(299,0,0,0)$, $x_1 + 1.5 - 1.6(x_2 + x_3)/299$ and $x_2 + 1.5 -1.6(x_3 + x_4)/299$ will be positive while $x_3 + 1.5 -1.6(x_1 + x_4)/299$ and $x_4 + 1.5 -1.6(x_1 + x_2)/299$ will be negative. By symmetry, condition 1 holds for the vector $(0,299,0,0)$, $(0,0,299,0)$, and $(0,0,0,299)$ as well.

For all other vectors, the maximum magnitude of the sum of two coordinates is $63$. Since $\frac{63}{299} < \frac{1}{4}$, the sign of the linear form is determined by whether the coordinate with weight $1$ has value at least $-1$ or has value at most $-2$. Since all of the other vectors have two coordinates which are at least $-1$ and two coordinates which are at most $-2$, condition 1 holds for all of vectors in $\{v_b\}$, as needed.

For the second condition we can take the following distribution for the linear form $l = x_1 + 1.5 - 1.6(x_2 + x_3)/299$:
\begin{enumerate}
\item Take $(299,0,0,0)$ with probability $\frac{15}{4500}$
\item Take $(-1,-1,-7,-7)$ with probability $\frac{1196}{4500}$, take $(-1,-7,-1,-7)$ with probability $\frac{1196}{4500}$, and take $(-1,-7,-7,-1)$ with probability $\frac{1196}{4500}$
\item Take $(-1,64,-2,-2)$ with probability $\frac{299}{4500}$, take $(-1,-2,64,-2)$ with probability $\frac{299}{4500}$, and take $(-1,-2,-2,64)$ with probability $\frac{299}{4500}$
\end{enumerate}
With this distribution, we have the following expectation values:
\begin{enumerate}
\item $E[v_1] = \frac{1}{4500}(299*15 - 1196*3 - 299*3) = 0$
\item $E[v_2] = \frac{1}{4500}(-1196*(1+7+7) + 299*(64-2-2)) = 0$
\item $E[{v_1}{v_2}]  = \frac{1}{4500}(1196*1*(1+7+7) - 299*1*(64-2-2)) = 0$
\item $E[{v_2}{v_3}]  = \frac{1}{4500}(1196*(7+7+49) - 299*(128+128-4)) = 0$
\item $E[{v^2_1}]  = \frac{1}{4500}(15*(299)^2 + 1196*3 + 299*3) = \frac{1345500}{4500} = 299$
\item $E[{v^2_2}]  = \frac{1}{4500}(1196*(1+49+49) + 299*((64)^2+4+4)) = \frac{1345500}{4500} = 299$
\end{enumerate}
By symmetry, the remaining expectation values match as well and we can take similar distributions for the other linear forms.
\end{proof}
\begin{remark}
We give some intution for how we found this core in Appendix \ref{findingcoreappendix}
\end{remark}
\subsection{Constraints and LTFs}\label{constraintsubsection}
In order to obtain a perfect integrality gap instance from this core, we have to fix the following two problems:
\begin{enumerate}
\item A priori, we can take any vector $v$, not just the vectors in $V$
\item Our variables are not boolean.
\end{enumerate}
To fix these problems, we will add constraints to our LTFs in a way such that if the constraints are not satisfied, then we automatically satisfy precisely $\frac{1}{2}$ of our LTFs.
\begin{definition}
Given a linear form $l$, let $Z(l) = \{x: l(x) = 0\}$
\end{definition}
\begin{proposition}
Let $l_{constraint}$ and $l_{remainder}$ be two linear forms. For all sufficiently large $B$, if we take the linear forms $l'_1 = {B}l_{constraint} + l_{remainder}$ and $l'_2 = -{B}l_{constraint} + l_{remainder}$ then 
\begin{enumerate}
\item If $x \in Z(l_{constraint})$ then $sign(l'_1(x)) = sign(l'_2(x)) = sign(l_{remainder}(x))$.
\item If $x \notin Z(l_{constraint})$ then $l'_1(x)l_2'(x) < 0$.
\end{enumerate}
\end{proposition}
\begin{proof}
The first statement is trivial. For the second statement, let $a = \min{\{|l_{constraint}(x)|: l_{constraint}(x) \neq 0}\}$ and let $b = \max{\{|l_{remainder(x)}|\}}$. Now note that as long as $B > \frac{b}{a}$ and $x \notin Z(l_{constraint})$, the sign of $l'_1(x)$ and $l'_2(x)$ is completely determined by the sign of $l_{constraint}(x)$
\end{proof}
Using this proposition, if we take two copies of each linear form $l_i$ and add $\pm{B}{l_{constraint}}$ to these copies for a sufficiently large $B$ then we will automatically satisfy half of our constraints unless $l_{constraint}(x) = 0$, in which case the answer is unchanged. This allows us to enforce the constraint that $x \in Z_{l_{constraint}}$ which is quite powerful.
\begin{example}
If we take $l_{constraint} = x_1 + x_2 + x_3 - 2x_4 - x_5$ where $x_1,x_2,x_3,x_4,x_5 \in \{-1,+1\}$ then $Z(l_{constraint}) = \{x:x_4 = sign(x_1 + x_2 + x_3), x_5 = {x_1}{x_2}{x_3}\}$
\end{example}
As shown by the following proposition, we can easily take the AND of multiple constraints.
\begin{proposition}
For any linear forms $l_1$ and $l_2$, for all sufficiently large constants $B$, $Z({B}l_1 + l_2) = Z(l_1) \cap Z(l_2)$
\end{proposition}
\begin{remark}
We have to be careful when adding constraints with new variables because in order to show that the resulting LTFs are a perfect integrality gap instance, we will have to give expectation values and pairwise expectation values for the new variables.
\end{remark}
\begin{remark}
The LTFs $sign({B}l_{constraint}(x) + l_{remainder}(x))$ and $sign(-{B}l_{constraint}(x) + l_{remainder}(x))$ may not have the same form. This is why we will need additional ideas to find a perfect integrality gap instance where all of the $\{f_i\}$ have the same form $P$ which is a balanced LTF.
\end{remark}
\subsection{Specifying potential solution vectors}\label{specificationsubsection}
In this subsection, we show how to use constraints to restrict the set of possible vectors to an arbitrary set of vectors $V = \{v_1,\cdots,v_m\}$.
\begin{lemma}\label{specificationlemma}
Given a set of linear forms $\{l_a\}$, a set of possible solution vectors $V$, and values $\{c_i\},\{c_{ii}\},\{c_{ij}\}$ such that for all $a$, there is a distrubution $D_a$ such that: 
\begin{enumerate}
\item $D_a$ is supported on the set $\{v: v \in V, l_a(v) > 0\}$
\item $\forall i, E_{v \in D_a}[v_i] = c_i$
\item $\forall i, E_{v \in D_a}[v^2_i] = c_{ii}$
\item $\forall i,j, E_{v \in D_a}[{v_i}{v_j}] = c_{ij}$
\end{enumerate}
we can construct sets of linear forms $\{l'_{a1}\} \cup \{l'_{a2}\}$, a set of solution vectors $V'$, and values $\{c'_i\},\{c'_{ii}\},\{c'_{ij}\}$ such that 
\begin{enumerate}
\item If $v' \notin V'$ then for all $a$, $l'_{a1}(v')l'_{a2}(v') < 0$.
\item For all $a,\sigma$ and all $v_b \in V$ there is a vector $v'_{b\sigma} \in V'$ such that $l'_{a1}(v'_{b\sigma}) = l'_{a2}(v'_{b\sigma}) = l_a(v_b)$.
\item For all $a$ there exists a distribution $D'_a$ such that
\begin{enumerate}
\item $D'_a$ is supported on the set $\{v': v' \in V', l'_{a1}(v') > 0\}$
\item $\forall i, E_{v' \in D'_a}[v'_i] = c'_i$
\item $\forall i, E_{v' \in D'_a}[{v'}^2_i] = c'_{ii}$
\item $\forall i,j, E_{v' \in D'_a}[{v'_i}{v'_j}] = c'_{ij}$
\end{enumerate}
\end{enumerate}
Moreover, if $V$ and the values $\{c_i\},\{c_{ii}\},\{c_{ij}\}$ are symmetric under permutations of the input variables and the linear forms $\{l_a\}$ are the same as some linear form $l$ up to permuting the variables then we may take $\{l'_{a1}\} \cup \{l'_{a2}\}$ so that all of the $\{l'_{a1}\}$ are the same as some linear form $l'_1$ up to permuting the variables and all of the $\{l'_{a2}\}$ are the same as sone linear form $l'_2$ up to permuting the variables.
\end{lemma}
\begin{proof}
The intution is as follows. We specify the set of vectors $V$ as possible solution vectors. We then use permutation gadgets to ensure that our final vector $v$ is one of the these vectors but we don't know which one.
\begin{definition}
We define a permutation gadget $P(x_1,\cdots,x_m)$ on variables $x_1,\cdots,x_m$ to consist of the following variables and constraints. For the variables, we have
\begin{enumerate}
\item Initial variables $x_1,\cdots,x_m$ which are integers in the range $[-B,B]$ for some bound $B$.
\item Output variables $y_1,\cdots,y_m$ which are integers in the range $[-B,B]$
\item Permutation indicators $\{p_{ij}: i,j \in [1,m]\}$ which are either $0$ or $1$. We want $p_{ij} = 1$ if $i$ is mapped to $j$ and $p_{ij} = 0$ otherwise.
\item Variables $\{d^{+}_{ij},d^{-}_{ij}: i,j \in [1,m]\}$ which are integers in the range $[-2B,2B]$. We want $\frac{d^{+}_{ij} - d^{-}_{ij}}{2} = y_j - x_i$
\end{enumerate}
For the constraints, we have
\begin{enumerate}
\item $\forall i, \sum_{j}{p_{ij}} = 1$
\item $\forall j, \sum_{i}{p_{ij}} = 1$
\item $\forall i,j, \frac{d^{+}_{ij} - d^{-}_{ij}}{2} = y_j - x_i$
\item $\forall i,j, d^{+}_{ij} + d^{-}_{ij} = 4Bp_{ij}$. This implies that $\forall i,j, d^{+}_{ij} = d^{-}_{ij} = 2B$ whenever $p_{ij} = 1$
\end{enumerate}
\end{definition}
\begin{proposition}
If the constraints are satisfied then $y_1,\cdots,y_m$ must be a permutation of $x_1,\cdots,x_m$
\end{proposition}
Similarly, we can construct a permutation gadget for vectors
\begin{definition}
We define a permutation gadget $P(v_1,\cdots,v_m)$ on vectors $v_1,\cdots,v_m$ to consist of the following variables and constraints. For the variables, we have
\begin{enumerate}
\item Initial variables $(v_{11},\cdots,v_{1n}),\cdots,(v_{m1},\cdots,v_{mn})$ which are integers in the range $[-B,B]$ for some bound $B$.
\item Output variables $(w_{11},\cdots,w_{1n}),\cdots,(w_{m1},\cdots,w_{mn})$ which are integers in the range $[-B,B]$
\item Permutation indicators $\{p_{ij}: i,j \in [1,m]\}$ which are either $0$ or $1$. We want $p_{ij} = 1$ if $i$ is mapped to $j$ and $p_{ij} = 0$ otherwise.
\item Variables $\{d^{+}_{ijk},d^{-}_{ijk}: i,j \in [1,m], k \in [1,n]\}$ which are integers in the range $[-2B,2B]$. We want $\frac{d^{+}_{ijk} - d^{-}_{ijk}}{2}= v_{jk} - w_{ik}$
\end{enumerate}
For the constraints, we have
\begin{enumerate}
\item $\forall i, \sum_{j}{p_{ij}} = 1$
\item $\forall j, \sum_{i}{p_{ij}} = 1$
\item $\forall i,j,k, \frac{d^{+}_{ijk} - d^{-}_{ijk}}{2} = v_{jk} - w_{ik}$
\item $\forall i,j,k, d^{+}_{ijk} + d^{-}_{ijk} = 4Bp_{ij}$. This implies that $\forall i,j,k, d^{+}_{ijk} = d^{-}_{ijk} = 2B$ whenever $p_{ij} = 1$
\end{enumerate}
\end{definition}
\begin{proposition}
If the constraints are satisfied then $w_1,\cdots,w_m$ must be a permutation of $v_1,\cdots,v_m$
\end{proposition}
We now describe our construction.
\begin{enumerate}
\item We take $V = \{v_1,\cdots,v_m\}$ to be the set of possible solution vectors
\item We take the permutation gadget $P(v_1,\cdots,v_m)$
\item We take a second permutation gadget $P(v'_1,\cdots,v'_{m})$ where $v'_1,\cdots,v'_{m} = w_1,\cdots,w_{m}$ are the output vectors of the first permutation gadget.
\item We take a third permutation gadget $P(v''_1,\cdots,v''_{m})$ where $v''_1,\cdots,v''_{m} = w'_1,\cdots,w'_{m}$ are the output vectors of the second permutation gadget.
\item We take $w''_1,\cdots,w''_{m}$ to be the output vectors of the third permutation gadget.
\item To obtain the non-constraint part of the linear forms $\{l'_{a1}\}$ and $\{l'_{a2}\}$, we apply the linear form $l_a$ to $w''_{1}$
\end{enumerate}
\begin{figure}[ht]
\centerline{\includegraphics[height=8cm]{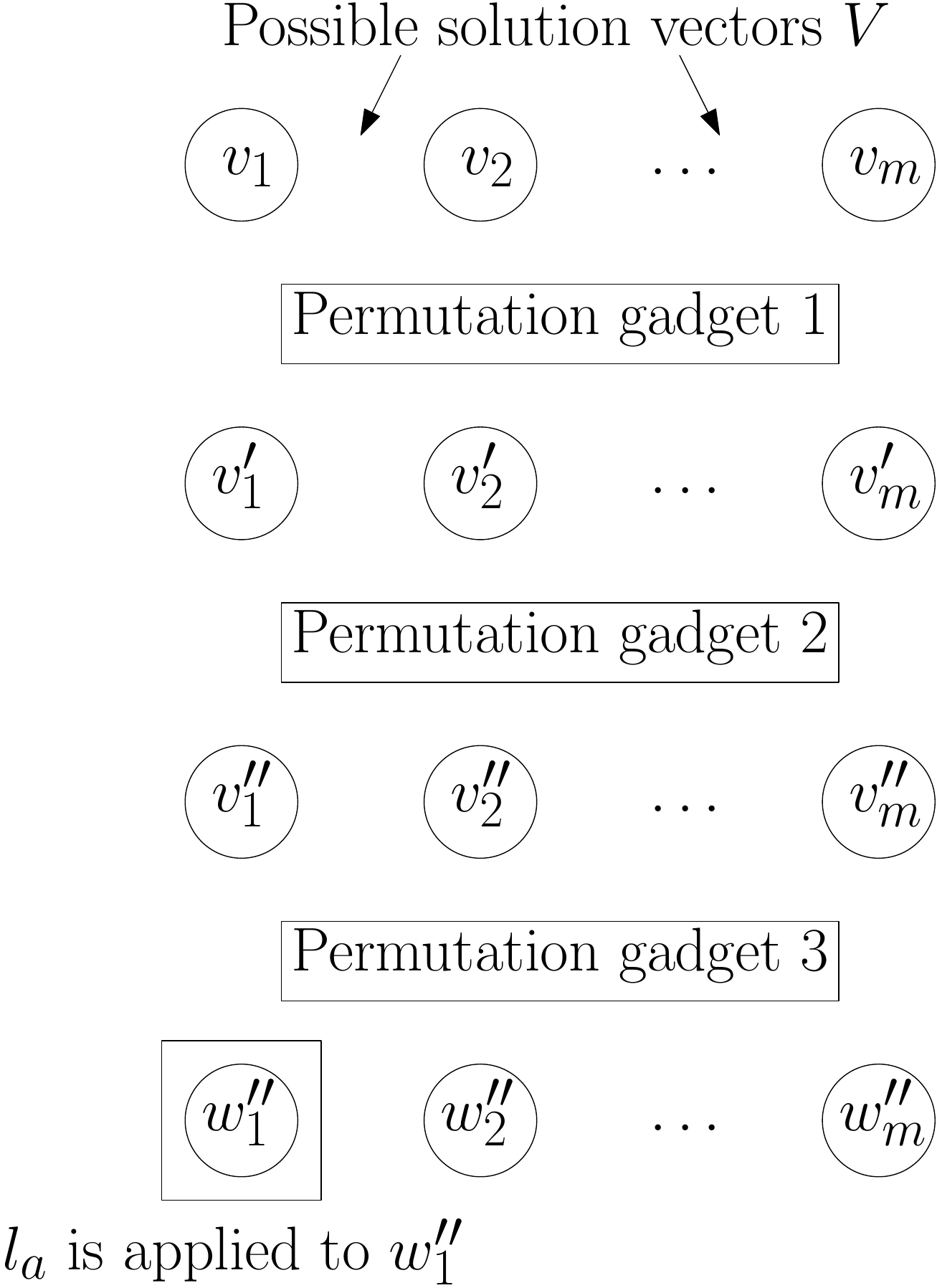}}
\caption{This figure illustrates our construction of the linear forms $l'_{a1}$ and $l'_{a2}$}
\label{permutationgadgetsfigure}
\end{figure}
We take $D'_{a}$ to be the following distribution. We start with the distribution $D_a$ for $w''_{1}$. Whenever we have that $w''_{1} = v_i$,  we take the uniform distribution over all triples of permutations $\sigma,\sigma',\sigma''$ such that $\sigma'' \circ \sigma' \circ \sigma(i) = 1$. The following lemma implies that we can find the values $c'$ as needed.
\begin{lemma}\label{pairwiseexpectationlemma}
For the distribution $D'_a$, all expectation values and pairwise expectation values depend only on the values $\{E[w''_{1i}] = c_i\}$, $\{E[(w''_{1i})^2] = c_{ii}\}$, and $\{E[w''_{1i}w''_{1j}] = c_{ij}\}$.
\end{lemma}
\begin{proof}
We first observe that we do not need to consider the variables $d^{+}_{ijk}$ and $d^{-}_{ijk}$. To see this, note that we can make the substitutions 
$d^{+}_{ijk} = 2Bp_{ij} + v_{jk} - w_{ik}$ and $d^{-}_{ijk} = 2Bp_{ij} - v_{jk} + w_{ik}$ and use linearity. Following similar logic, we do not need to consider the variables $d'^{+}_{ijk}$, $d'^{-}_{ijk}$, $d''^{+}_{ijk}$, or $d''^{-}_{ijk}$ either.

To analyze the remaining variables, we take $a_i = E_{v \in \{v_1,\dots,v_m\}}[v_i]$ and $a_{ij} = E_{v \in \{v_1,\dots,v_m\}}[{v_i}{v_j}]$. Looking at the expected values, we have that
\begin{enumerate}
\item The values $\{v_{ik}\}$ are fixed
\item $\forall i,j, E[p_{ij}] = E[p'_{ij}] = E[p''_{ij}] = \frac{1}{m}$
\item $\forall i,k, E[v'_{ik}] = E[v''_{ik}] = a_{k}$
\item For all $i \neq 1$ and all $k$, $E[w''_{ik}] = \frac{m}{m-1}a_k - \frac{c_k}{m-1}$ because $E[w''_{1k}] = c_k$ and we must have that $\sum_{i=1}^{k}{w''_{ik}} = m(a_k)$
\end{enumerate}
We now observe that the only pairs of variables which are not pairwise independent are pairs of permutation indicators in the same permutation gadget and pairs of coordinates from vectors in the same set $\{v'_1,\cdots,v'_m\}$, $\{v''_1,\cdots,v''_m\}$, or $\{w''_1,\cdots,w''_m\}$. For pairs of permutation indicators in the same permutation gadget, pairs of coordinates of vectors in the set $\{v'_1,\cdots,v'_m\}$, and pairs of coordinates of vectors in the set $\{v''_1,\cdots,v''_m\}$, the pairwise expectation values will be the same regardless of the distribution of $w''_1$. Thus, we just need to consider pairs of coordinates of vectors in the set $\{w''_1,\cdots,w''_m\}$ and we obtain the following expected values:
\begin{enumerate}
\item $E[w''_{1i}w''_{2j}] = \frac{m}{m-1}{a_j}c_{i} - \frac{c_{ij}}{m-1}$ because 
\[
\sum_{k = 2}^{m}{E[w''_{1i}w''_{kj}]} = m{a_j}E[w''_{1i}] - E[w''_{1i}w''_{1j}] = m{a_j}c_{i} - c_{ij}
\]
\item $E[w''_{2i}w''_{2j}] = \frac{m}{m-1}a_{ij} - \frac{c_{ij}}{m-1}$ because 
\[
\sum_{k = 2}^{m}{E[w''_{2i}w''_{2j}]} = \sum_{k = 1}^{m}{E[w''_{ki}w''_{kj}]} - E[w''_{1i}w''_{1j}] = ma_{ij} - c_{ij}
\]
\item $E[w''_{2i}w''_{3j}] = \frac{m^2{a_i}{a_j} - ma_j{c_i} - ma_i{c_j} - ma_{ij} + 2c_{ij}}{(m-1)(m-2)}$ because
\begin{align*}
\sum_{k = 3}^{n}{E[w''_{2i}w''_{kj}]} &= m{a_j}E[w''_{2i}] - E[w''_{2i}w''_{2j}] - E[w''_{2i}w''_{1j}] \\
&= \frac{m^2}{m-1}a_i{a_j} - \frac{m{c_i}a_j}{m-1} - \left(\frac{m}{m-1}a_{ij} - \frac{c_{ij}}{m-1}\right) - \left(\frac{m}{m-1}a_i{c_j} - \frac{c_{ij}}{m-1}\right) \\
&= \frac{m^2{a_i}{a_j} - ma_j{c_i} - ma_i{c_j} - ma_{ij} + 2c_{ij}}{m-1}
\end{align*}
\end{enumerate}
\end{proof}
To see the moreover part, we make the following observation. Observe that the set of constraints we are adding is symmetric under permutations of the $k$ indices (which corresponds to permutations of the input variables to the core). Thus, letting $l_{constraint}$ be a linear form enforcing the constraints, instead of adding and subtracting $l_{constraint}$ to each $l_a$ to obtain the linear forms $l'_{a1}$ and $l'_{a2}$, we can instead add and subtract $l_{constraint}$ to a single $l_a$ to obtain the corresponding $l'_{a1}$ and $l'_{a2}$ and then for all $a' \neq a$ we can apply the corresponding permutation of the $k$ indices which maps $l_a$ to $l_{a'}$ to $l'_{a1}$ and $l'_{a2}$ to obtain the linear forms $l'_{a'1}$ and $l'_{a'2}$ 
\end{proof}
\begin{remark}
We need 3 consecutive permutation gadgets so that almost all pairs of variables (excluding the $d^{+}_{ijk}$ and $d^{-}_{ijk}$ variables) will be pairwise independent. If we only had one permutation gadget, we would have to have that $p_{i1}$ is always $0$ whenever $v_i$ is not in the support of our distribution. Similarly, if we only had two permutation gadgets, we would have to have that $p_{ij}p'_{j1} = 0$ whenever $v_i$ is not in the support of our distribution.
\end{remark}
\subsection{Expressing variables in unary}\label{unarysubsection}
So far we have worked with variables $x_i$ which take integer values in some range $[a,b]$ where $a,b \in \mathbb{Z}$. We now describe how to replace these variables using $\pm{1}$ variables.
\begin{definition}
We define a variable $x_{one}$ which is always supposed to be $1$. In particular, in all distributions we take $E[x_{one}] = 1$ and we take $E[x_{one}x_i] = E[x_i]$ for every variable $x_i$
\end{definition}
\begin{remark}
If every assignment where $x_{one} = 1$ satisfies exactly half of the balanced LTFs then by symmetry, every assignment satisfies exactly half of the balanced LTFs. Thus, without loss of generality we can assume that $x_{one}$ is always 1.
\end{remark}
\begin{lemma}\label{unarylemma}
Given variables $x_i$ which takes integer values in some range $[a_i,b_i]$ where $a,b \in \mathbb{Z}$, we can replace each $x_i$ with $\frac{b_i+a_i}{2}x_{one} + \frac{1}{2}\sum_{k=1}^{b_i-a_i}{y_{ik}}$ where the $\{y_{ik}\}$ are $\pm{1}$ variables. Moreover, given a distribution $D$ such that 
\begin{enumerate}
\item $\forall i, E_D[x_i] = c_i$
\item $\forall i, E_D[x^2_i] = c_{ii}$
\item $\forall i < j, E_D[x_{i}x_{j}] = c_{ij}$
\end{enumerate}
there is a distribution $D'$ on the variables $\{y_{ik}\}$ such that
\begin{enumerate}
\item $\forall i,\forall k, E_{D'}[y_{ik}] = \frac{2c_i - b_i+a_i}{b_i - a_i}$
\item $\forall i, \forall k, E_{D'}[y^2_{ik}] = 1$
\item $\forall i, \forall k_1,k_2, E_{D'}[y_{ik_1}y_{ik_2}] = \frac{2c_{ii} - 2c_i - \frac{(b_i+a_i)^2}{2} + b_i + a_i - \frac{(b_i - a_i)}{2}}{(b_i - a_i)(b_i - a_i - 1)}$
\item $\forall i < j, \forall k_1,k_2, E_{D'}[y_{ik_1}y_{jk_2}] = \frac{4c_{ij} - 2(b_i + a_i)c_j - 2(b_j + a_j)c_i + (b_i + a_i)(b_j + a_i)}{(b_i - a_i)(b_j - a_j)}$
\end{enumerate}
\end{lemma}
\begin{proof}
For given values of $\{x_i\}$, for each $i$ we randomly choose $2x_i - b_i - a_i$ of the variables $y_{ik}$ to be $1$ and $2b_i - 2x_i$ of the variables $y_{ik}$ to be $-1$. Applying this to all of the possible values $\{x_i\}$ in $D$, we obtain the distribution $D'$. We now make the following computations (where we replace $x_{one}$ by $1$ throughout).
\begin{enumerate}
\item 
\[
x_i = \frac{b_i+a_i}{2} + \frac{1}{2}\sum_{k=1}^{b_i-a_i}{y_{ik}}
\]
By symmetry, for a given $x_i$, 
\[
E[y_{ik}] = \frac{2x_i - b_i+a_i}{b_i - a_i}
\]
Taking the expected value over $D$, 
\[
E_{D'}[y_{ik}] = \frac{2c_i - b_i+a_i}{b_i - a_i}
\]
\item 
\begin{align*}
x^2_i &= \left(\frac{b_i+a_i}{2} + \frac{1}{2}\sum_{k=1}^{b_i-a_i}{y_{ik}}\right)^2 \\
&= \frac{(b_i+a_i)^2}{4} + \frac{1}{2}\sum_{k=1}^{b_i-a_i}{y_{ik}} + \frac{1}{4}\sum_{k=1}^{b_i-a_i}{y^2_{ik}} + 
\frac{1}{2}\sum_{k_1,k_2: k_1 < k_2, k_1,k_2 \in [1,b_i-a_i]}{y_{i{k_1}}y_{i{k_2}}} \\
&= \frac{(b_i+a_i)^2}{4} + x_i - \frac{b_i+a_i}{2} + \frac{(b_i - a_i)}{4} 
+ \frac{1}{2}\sum_{k_1,k_2: k_1 < k_2, k_1,k_2 \in [1,b_i-a_i]}{y_{i{k_1}}y_{i{k_2}}}.
\end{align*} 
By symmetry, for a given value of $x_i$, 
\[
E[y_{i{k_1}}y_{i{k_2}}] = \frac{2x^2_i - 2x_i - \frac{(b_i+a_i)^2}{2} + b_i + a_i - \frac{(b_i - a_i)}{2}}{(b_i - a_i)(b_i - a_i - 1)}
\]
Taking the expected value over $D$,
\[
E_{D'}[y_{i{k_1}}y_{i{k_2}}] = \frac{2c_{ii} - 2c_i - \frac{(b_i+a_i)^2}{2} + b_i + a_i - \frac{(b_i - a_i)}{2}}{(b_i - a_i)(b_i - a_i - 1)}
\]
\item 
\begin{align*}
{x_i}{x_j} &= \left(\frac{b_i+a_i}{2} + \frac{1}{2}\sum_{k_1=1}^{b_i-a_i}{y_{i{k_1}}}\right)\left(\frac{b_j+a_j}{2} + \frac{1}{2}\sum_{k_2=1}^{b_j-a_j}{y_{j{k_2}}}\right) \\
&= \frac{(b_i + a_i)(b_j + a_i)}{4} + \frac{(a_i + b_i)}{4}\sum_{k=1}^{b_j-a_j}{y_{jk}} + \frac{(a_j + b_j)}{4}\sum_{k=1}^{b_i-a_i}{y_{ik}} + 
\frac{1}{4}\sum_{k_1=1}^{b_i-a_i}\sum_{k_2=1}^{b_j-a_j}{{y_{ik_1}y_{jk_2}}} \\
&= \frac{(b_i + a_i)(b_j + a_i)}{4} + \frac{(b_i + a_i)(2x_j - a_j - b_j)}{4} + \frac{(b_j + a_j)(2x_i - a_i - b_i)}{4} \\ 
&+ \frac{1}{4}\sum_{k_1=1}^{b_i-a_i}\sum_{k_2=1}^{b_j-a_j}{{y_{ik_1}y_{jk_2}}} \\
&= \frac{(b_i + a_i)x_j}{2} + \frac{(b_j + a_j)x_i}{2} - \frac{(b_i + a_i)(b_j + a_i)}{4} + \frac{1}{4}\sum_{k_1=1}^{b_i-a_i}\sum_{k_2=1}^{b_j-a_j}{{y_{ik_1}y_{jk_2}}}
\end{align*}
By symmetry, for given values of $x_i,x_j$, 
\[
E[y_{ik_1}y_{jk_2}] = \frac{4{x_i}{x_j} - 2(b_i + a_i)x_j - 2(b_j + a_j)x_i + (b_i + a_i)(b_j + a_i)}{(b_i - a_i)(b_j - a_j)}
\]
Taking the expected value over $D$, 
\[
E_{D'}[y_{ik_1}y_{jk_2}] = \frac{4c_{ij} - 2(b_i + a_i)c_j - 2(b_j + a_j)c_i + (b_i + a_i)(b_j + a_i)}{(b_i - a_i)(b_j - a_j)}
\]
\end{enumerate}
\end{proof}
\subsection{A perfect integrality gap instance with two LTFs}
Putting together the ideas we have so far, we can find a perfect integrality gap instance which consists of two different balanced LTFs
\begin{theorem}\label{twoLTFtheorem}
There exist two balanced linear forms $l_1,l_2: \{-1,+1\}^k$ and a perfect integrality gap instance $\{f_i\}$ where each $f_i$ is on the same set of variables $x_1,\dots,x_k$ and has the form $sign(l_1)$ or $sign(l_2)$. In fact, we may take $l_1$ and $l_2$ to be perfectly balanced (see Definition \ref{perfectlybalanceddefinition}).
\end{theorem}
\begin{proof}
By Lemma \ref{corelemma}, there is a set of linear forms $\{l_a\}$, a set of possible solution vectors $V$, and values $\{c_i\},\{c_{ii}\},\{c_{ij}\}$ such that
\begin{enumerate}
\item $\forall v \in V, Pr_{l \in \{l_a\}}[sign(l(v)) > 0] = \frac{1}{2}$
\item For all $a$ there exists a distribution $D_a$ such that
\begin{enumerate}
\item $D_a$ is supported on the set $\{v: v \in V, l_a(v) > 0\}$
\item $\forall i, E_{v \in D_a}[v_i] = c_i$
\item $\forall i, E_{v \in D_a}[v^2_i] = c_{ii}$
\item $\forall i,j, E_{v \in D_a}[{v_i}{v_j}] = c_{ij}$
\end{enumerate}
\end{enumerate}
In fact, we may take $V$ and the values $\{c_i\},\{c_{ii}\},\{c_{ij}\}$ to be symmetric under permutations of the input variables and have that all of the $\{l_a\}$ are the same as some linear form $l$ up to permuting the input variables.

Using Lemma \ref{specificationlemma}, we can construct sets of linear forms $\{l'_{a1}\} \cup \{l'_{a2}\}$, a set of solution vectors $V'$, and values $\{c'_i\},\{c'_{ii}\},\{c'_{ij}\}$ such that 
\begin{enumerate}
\item If $v' \notin V'$ then $l'_{a1}(v')l'_{a2}(v') < 0$ for all $a$.
\item For all $a,\sigma$ and all $v_b \in V$ there exists a vector $v'_{b\sigma} \in V'$ such that $l'_{a1}(v'_{b\sigma}) = l'_{a2}(v'_{b\sigma}) = l_a(v_b)$.
\item For all $a$ there exists a distribution $D'_a$ such that
\begin{enumerate}
\item $D'_a$ is supported on the set $\{v': v' \in V', l'_{a1}(v') > 0\}$
\item $\forall i, E_{v' \in D'_a}[v'_i] = c'_i$
\item $\forall i, E_{v' \in D'_a}[{v'}^2_i] = c'_{ii}$
\item $\forall i,j, E_{v' \in D'_a}[{v'_i}{v'_j}] = c'_{ij}$
\end{enumerate}
\end{enumerate}
Moreover, we may take the linear forms $\{l'_{a1}\} \cup \{l'_{a2}\}$ so that all of the linear forms $\{l'_{a1}\}$ are the same as some linear form $l'_1$ up to permutations of the variables and all of the linear forms $\{l'_{a2}\}$ are the same as some linear form $l'_2$.

The first two conditions imply that for all $v'$, exactly half of $\{l'_{a1}(v')\} \cup \{l'_{a2}(v')\}$ are positive. To see that we can take $l'_1$ to have boolean variables and be perfectly balanced, we use the following idea. Recall that we obtained the linear forms $\{l'_{a1}(v')\} \cup \{l'_{a2}(v')\}$ by first finding a single $l'_{a1}$ and $l'_{a2}$ and then permuting the original variables of the core (which correspond to the $k$ indices in Lemma \ref{specificationlemma}) to obtain the linear forms $l'_{a'1}$ and $l'_{a'2}$ for all $a' \neq a$. We adjust this procedure to first apply transformations to $l'_{a1}$ and $l'_{a2}$ and then permute the original variables of the core (which correspond to the $k$ indices in Lemma \ref{specificationlemma}) as before to obtain the linear forms $l'_{a'1}$ and $l'_{a'2}$ for all $a' \neq a$. In particular, we use Lemma \ref{unarylemma} to make $l'_{a1}$ and $l'_{a2}$ have boolean variables. We then add dummy variables to make $k$ be a power of $2$ and use Lemma \ref{balancinglemma} to make $l'_{a1}$ and $l'_{a2}$ perfectly balanced. In this way, we can make the linear forms $l'_1$ and $l'_2$ have boolean variables and be perfectly balanced, giving us a perfect integrality gap instance.
\end{proof}
\subsection{Finding a single balanced LTF which is unique games hard to approximate}\label{mergingsubsection}
We now describe how to use this perfect integrality gap instance with two different balanced LTFs to find a perfect integrality gap instance with a single balanced LTF.
\begin{definition}\label{perfectlybalanceddefinition}
We say that a linear form $l$ on $k$ $\pm{1}$ variables is perfectly balanced if 
\[
\forall j \in [1,k-1], Pr_{x \in \{x:\sum_{i}{x_i} = 2j-k\}}[l(x) > 0] = \frac{1}{2}
\]
\end{definition}
\begin{example}
The linear form $l(x) = 2x_1 + x_2 - x_3 - x_4$ is perfectly balanced.
\end{example}
\begin{proposition}
A linear form can only be perfectly balanced if $k$ is a power of $2$.
\end{proposition}
\begin{lemma}\label{dualsimulationlemma}
If $l_1 = \sum_{i=1}^{k}{w_{i}x_{i}}$ and $l_2 = \sum_{i=1}^{k}{w'_{i}x_{i}}$ are two perfectly balanced linear forms on $k$ variables then there exists a linear form $l_3 = \sum_{i=1}^{k}{\sum_{j=1}^{k}}{w_{ij}x_{ij}}$ on $k^2$ variables which is perfectly balanced and has the following properties: 
\begin{enumerate}
\item $\forall i, \sum_{j=1}^{k}{w_{ij}} = w_i$
\item $\forall j, \sum_{i=1}^{k}{w_{ij}} = w'_j$
\item If $\{x_{ij}\}$ have values such that there exist $i,j,j'$ such that $x_{ij} \neq x_{ij'}$ then 
\[
E_{\sigma \in S_k}[sign(l_3(\{x_{i\sigma(j)})\})] = 0
\]
\item If $\{x_{ij}\}$ have values such that there exist $i,i',j,$ such that $x_{ij} \neq x_{i'j}$ then 
\[
E_{\sigma \in S_k}[sign(l_3(\{x_{\sigma(i)j})\})] = 0
\]
\end{enumerate}
\end{lemma}
\begin{proof}
To obtain this linear form, we start by finding a linear form $l$ which obeys the first two statements which can be done by solving a system of linear equations. We will then add terms of the form $B(x_{ij} - x_{ij'} - x_{i'j} + x_{i'j'})$ to $l$. The idea is that this effectively adds the constraint 
$x_{ij} - x_{ij'} - x_{i'j} + x_{i'j'} = 0$. This constraint is satisfied for all $i < i', j < j'$ if we are either constant along rows or constant along columns. Otherwise, there will be such a constraint which is violated.

However, here we cannot take a set of constraints and their negations because this would give two different LTFs while we are trying to only have one LTF. Instead, we must ensure that when a constraint is violated and we average over the permutations, we get each sign with equal probability. We can do this as follows. We take
\[
l_3 = l + \sum_{y,a,b \in \{0,1\}^{{log}_2(k)}:a \cdot z_y \bmod 2 = 0, b \cdot z_y \bmod 2 = 0}{B_{yab}(x_{ab} - x_{ab'} - x_{a'b} + x_{a'b'})}
\]
where $z_y$ is an aribtrary vector such that $y \cdot z_y \bmod 2 = 1$, $a' = a \oplus y$, $b' = b \oplus y$, and the $B_{yab}$ are exponentially decreasing constants (which are still much larger than the weights $\{w_i\}$ and $\{w'_i\}$).

If we are given values of $\{x_{ij}\}$ which violate these constraints then let $y'$ be the first $y$ such that a constraint $x_{ab} - x_{ab'} - x_{a'b} + x_{a'b'}$ where $a' = a \oplus y$ and $b' = b \oplus y$ is violated. Now take $\sigma$ to be the permutation $\sigma(x) = x \oplus y'$. Observe that regardless of whether we apply $\sigma$ to the rows or the columns, we change the sign of all constraints $x_{ab} - x_{ab'} - x_{a'b} + x_{a'b'} = 0$ where $a' = a \oplus y'$ and $b' = b \oplus y'$. Moreover, for all $y$, we keep the set of constraints $\{x_{ab} - x_{ab'} - x_{a'b} + x_{a'b'}, a' = a \oplus y, b' = b \oplus y\}$ the same. Thus, for earlier $y$ these constraints will still all be satisfied. For later $y$, we will permute which constraints are satisfied but this does not matter because they all have smaller coefficients. Thus, if any of these constraints are violated, both signs are equally likely when we average over permutations of the rows or average over permutations of the columns.
\begin{lemma}
If $\{x_{ij}\}$ satisfy the constraints then either $\forall i,j,j' \in [1,k], x_{ij} = x_{ij'}$ or $\forall i,i',j \in [1,k], x_{ij} = x_{i'j}$
\end{lemma}
\begin{proof}
Let $A = \{i: x_{ii} = -1\}$ and let $B = \{i: x_{ii} = 1\}$. Observe that for all $i,j$ we have the constraint $x_{ii} - x_{ij} - x_{ji} + x_{jj} = 0$. This implies that whenever $i,j \in A$, $x_{ij} = -1$ and whenever $i,j \in B$, $x_{ij} = 1$.

If $A = \emptyset$ or $B = \emptyset$ the result is now trivial. Otherwise, choose an $i \in A$ and a $j \in B$ and observe the following:
\begin{enumerate}
\item For all $y$ such that $i \oplus y \in A, j \oplus y \in A$, $x_{i(j \oplus y)} = x_{(i \oplus y)(j \oplus y)} = -1$ so $x_{(i \oplus y)j} = x_{ij}$. Similarly, $x_{(j \oplus y)i} = x_{(j \oplus y)(i \oplus y)} = -1$ so $x_{j(i \oplus y)} = x_{ji}$
\item For all $y$ such that $i \oplus y \in B, j \oplus y \in B$, $x_{(i \oplus y)j} = x_{(i \oplus y)(j \oplus y)} = 1$ so $x_{i(j \oplus y)} = x_{ij}$. Similarly, $x_{j(i \oplus y)} = x_{(j \oplus y)(i \oplus y)} = 1$ so $x_{(j \oplus y)i} = x_{ji}$
\item For all $y$ such that $i \oplus y \in B, j \oplus y \in A$, $x_{i(j \oplus y)} = -1$ and $x_{(i \oplus y)j} = 1$ so $x_{(i \oplus y)(j \oplus y)} = -x_{ij}$. Similarly, $x_{(j \oplus y)i} = -1$ and $x_{j(i \oplus y)} = 1$ so $x_{(j \oplus y)(i \oplus y)} = -x_{ji}$
\item For all $y$ such that $i \oplus y \in A, j \oplus y \in B$, $x_{(i \oplus y)(j \oplus y)} = x_{ij}$ and $x_{(j \oplus y)(i \oplus y)} = x_{ji}$. To see this, note that there are $|B|$ $y_2$ such that $i \oplus y_2 \in B$ and there are $|A|$ $y_2$ such that $j \oplus y_2 \in A$, so there must be at least one $y_2$ such that $i \oplus y_2 \in B, j \oplus y_2 \in A$. Now observe that 
\[
x_{(i \oplus y)(j \oplus y)} = -x_{(i \oplus y \oplus (y \oplus y_2))(j \oplus y \oplus (y \oplus y_2))} = -x_{(i \oplus y_2)(j \oplus y_2)} = x_{ij}
\]
Similarly, 
\[
x_{(j \oplus y)(i \oplus y)} = -x_{(j \oplus y \oplus (y \oplus y_2))(i \oplus y \oplus (y \oplus y_2))} = -x_{(j \oplus y_2)(i \oplus y_2)} = x_{ji}
\]
\end{enumerate}
There are now two cases to consider. Either $x_{ij} = -1$ or $x_{ij} = 1$. If $x_{ij} = -1$ then we expect to be constant along rows and if $x_{ij} = 1$ then we expect to be constant along columns. We confirm this as follows:
\begin{enumerate}
\item If $x_{ij} = -1$ then $x_{(i \oplus y)(j \oplus y)} = -1$ whenever $i \oplus y \in A$. To see this, note that if $i \oplus y \in A$ and $j \oplus y \in A$ then $x_{(i \oplus y)(j \oplus y)} = -1$ and if $i \oplus y \in A$ and $j \oplus y \in B$ then $x_{(i \oplus y)(j \oplus y)} = x_{ij} = -1$. This implies that $x_{ij} = -1$ whenever $i \in A$.

Given $i \in A,j \in B$, there exists a $y$ such that $i \oplus A \in B$ and $j \oplus y \in A$ and we have that 
\[
x_{ji} = -x_{(j \oplus y)(i \oplus y)} = 1
\]
This implies that $x_{ji} = 1$ whenever $j \in B$.
\item If $x_{ij} = 1$ then $x_{(i \oplus y)(j \oplus y)} = 1$ whenever $j \oplus y \in B$. To see this, note that if $i \oplus y \in B$ and $j \oplus y \in B$ then $x_{(i \oplus y)(j \oplus y)} = 1$ and if $i \oplus y \in A$ and $j \oplus y \in B$ then $x_{(i \oplus y)(j \oplus y)} = x_{ij} = 1$. This implies that $x_{ji} = 1$ whenever $j \in B$.

Given $i \in A,j \in B$, there exists a $y$ such that $i \oplus A \in B$ and $j \oplus y \in A$ and we have that 
\[
x_{ji} = -x_{(j \oplus y)(i \oplus y)} = -1
\]
This implies that $x_{ji} = -1$ whenever $i \in A$.
\end{enumerate} 
\end{proof}
Thus, if all of the constraints are satisfied then the $\{x_{ij}\}$ are either constant along rows or constant along columns. If the $\{x_{ij}\}$ are constant along rows then letting $x_i$ be the value of $x_{ij}$, 
\[
l_3(x) = \sum_{i=1}^{k}{\sum_{j=1}^{k}}{w_{ij}x_{ij}} = \sum_{i=1}^{k}{\left(\sum_{j=1}^{k}{w_{ij}}\right)x_i} = \sum_{i=1}^{k}{{w_i}{x_i}} = l_1(x_1,\dots,x_k)
\]
If the $\{x_{ij}\}$ are constant along columns then letting $x_j$ be the value of $x_{ij}$, 
\[
l_3(x) = \sum_{i=1}^{k}{\sum_{j=1}^{k}}{w_{ij}x_{ij}} = \sum_{j=1}^{k}{\left(\sum_{i=1}^{k}{w_{ij}}\right)x_j} = \sum_{j=1}^{k}{{w'_j}{x_j}} = l_2(x_1,\dots,x_k)
\]
We now check that $l_3$ is perfectly balanced. We can ignore the cases when the $\{x_{ij}\}$ are not constant along rows and not constant along columns because these cases average to $0$ under permutations. If the $\{x_{ij}\}$ are constant along rows and are not all the same then it is as if we have the variables $\{x_i\}$ for the linear form $l_1$. Since $l_1$ is perfectly balanced, if we permute the rows then we will obtain both signs with equal probability. Similarly, if the $\{x_{ij}\}$ are constant along columns and are not all the same then it is as if we have the variables $\{x_j\}$ for the linear form $l_2$. Since $l_2$ is perfectly balanced, if we permute the columns then we will obtain both signs with equal probability.
\end{proof}
\subsection{Making an LTF perfectly balanced}\label{balancingsubsection}
\begin{lemma}\label{balancinglemma}
Given a linear form $l$ on $k$ variables $\{x_1,\dots,x_k\}$ taking values in $\{-1,+1\}$ where $k$ is a power of $2$, there is a linear form $l'$ on $2k$ variables $\{x_1,\dots,x_k\} \cup \{y_1,\dots,y_k\}$ taking values in $\{-1,+1\}$ such that 
\begin{enumerate}
\item $l'$ is perfectly balanced.
\item If $y_i = -x_i$ for all $i \in [1,k]$ then $l'(x_1,\dots,x_k,y_1,\dots,y_k) = l(x_1,\dots,x_k)$
\end{enumerate}
\end{lemma}
\begin{proof}
We take $l'$ to have two components. The larger component will be nonzero as long as $y_i = x_i$ for some $i$ and will be $0$ if $y_i = -x_i$ for all $i$. This component will be chosen independently of $l$. The second component will be $l$. Thus, $l'$ has value $l(x_1,\dots,x_k)$ if $y_i = -x_i$ for all $i$ and otherwise the sign of $l'$ does not depend on $l$.

Observe that whether or not $l'$ is perfectly balanced depends only on the first component. To see this, note that the behavior of $l$ only matters in the layer where we have $k$ $-1$s and $k$ $1$s and we always have that half of the inputs to $l$ result in a positive sign and half of the inputs to $l$ result in a negative sign.

Taking the variables for the first component to be $\frac{x_i + y_i}{2}$, these variables can have 3 values, $-1$, $0$, or $1$. We generalize the definition for being perfectly balanced as follows:
\begin{definition}
We say that a linear form $l$ on $V$-valued variables (where $V$ is some finite set of values) is perfectly balanced if permuting the variables has equal probability to reslt in a positive or negative sign unless the variables all have the same value.
\end{definition}
Now we just need to find a construction of a single perfectly balanced $l$ on $2^j$ $\{-1,0,1\}$-valued variables for all $j \geq 0$. We can find such a construction inductively using the following lemma.
\begin{lemma}
Given a perfectly balanced linear form $l$ on $2^j$ $V$-valued variables, we can find a perfectly balanced linear form $l'$ on $2^{j+1}$ $V$-valued variables.
\end{lemma}
\begin{proof}
For each variable $x_i$ in $l$, take variables $x_i,y_i$ for $l'$. We take 
\[
l' = {B}l_{constraint}(\{y_i - x_i\}) + l(\{x_i\})
\]
where $l_{constraint}$ is an arbitrary linear form on $(V - V)$-valued variables which is only $0$ if all of its inputs are $0$ and we take $B$ to be a sufficiently large coefficient so that $l(\{x_i\})$ is negligible unless $l_{constraint}(\{y_i - x_i\}) = 0$

Observe that swapping all of the $x_i$ with the $y_i$ changes the sign of $l_{constraint}(\{y_i - x_i\})$ and thus $l'$ unless $x_i = y_i$ for all $i$. If $x_i = y_i$ for all $i$ then averaging over permutations of $\{x_i\}$ (and applying the same permutation to $\{y_i\}$) results in an equal probability of being positive or negative unless all of the variables are equal.
\end{proof}
\end{proof}
\subsection{Putting everything together}
We now put everything together to prove our main result.
\begin{theorem}
There exists a balanced linear form $l_3$ and a perfect integrality gap instance $\{f_i\}$ such that each $f_i$ has the form $sign(l_3)$.
\end{theorem}
\begin{corollary}
There exists a balanced LTF which is unique games hard to approximate.
\end{corollary}
\begin{proof}
By Theorem \ref{twoLTFtheorem}, there exist two perfectly balanced linear forms $l_1,l_2: \{-1,+1\}^k$ and a perfect integrality gap instance $\{f_i\}$ where each $f_i$ is on the same set of variables $x_1,\dots,x_k$ and has the form $sign(l_1)$ or $sign(l_2)$.

By Lemma \ref{dualsimulationlemma}, if $l_1 = \sum_{i=1}^{k}{w_{i}x_{i}}$ and $l_2 = \sum_{i=1}^{k}{w'_{i}x_{i}}$ then there exists a linear form $l_3 = \sum_{i=1}^{k}{\sum_{j=1}^{k}}{w_{ij}x_{ij}}$ on $k^2$ variables which is perfectly balanced and has the following properties: 
\begin{enumerate}
\item $\forall i, \sum_{j=1}^{k}{w_{ij}} = w_i$
\item $\forall j, \sum_{i=1}^{k}{w_{ij}} = w'_j$
\item If $\{x_{ij}\}$ have values such that there exist $i,j,j'$ such that $x_{ij} \neq x_{ij'}$ then 
\[
E_{\sigma \in S_k}[sign(l_3(\{x_{i\sigma(j)})\})] = 0
\]
\item If $\{x_{ij}\}$ have values such that there exist $i,i',j,$ such that $x_{ij} \neq x_{i'j}$ then 
\[
E_{\sigma \in S_k}[sign(l_3(\{x_{\sigma(i)j})\})] = 0
\]
\end{enumerate}
\begin{definition}
We define $X$ to be the matrix of variables $X_{ij} = x_{ij}$
\end{definition}
\begin{definition}
If $Y$ is a matrix of variables $Y_{ij} = y_{ij}$ then we define 
\[
l_3(Y) = l_3(y_{11},y_{12},\dots,y_{1k},y_{21},\dots,y_{k(k-1)}y_{kk})
\]
\end{definition}
We transform our perfect integrality gap instance into an integrality gap instance $\{f'_i\}$ where each $f'_{i}$ is on the same set of variables $x_{11},\dots,x_{kk}$ and has the form $sign(l_3)$ as follows. We replace each linear form $l_1(x_1,\dots,x_k)$ by $\{l_3(XP):P \text{ is a permutation matrix}\}$ and replace each linear form $l_2(x_1,\dots,x_k)$ by $\{l_3(PX):P \text{ is a permutation matrix}\}$. To obtain our new distributions $\{D'_a\}$, we take the old distributions $\{D_a\}$ and take $x_{ij} = x_i$ for all $i,j$.

We now observe that 
\begin{enumerate}
\item If $x_{ij} = x_i$ for all $i,j$ then for all permutation matrices $P$, $l_3(XP) = l_1(x_1,x_2,\dots,x_k)$ and $l_3(PX) = l_2(x_1,\dots,x_k)$. Thus, if $x_{ij} = x_i$ for all $i,j$ then our transformed instance behaves exactly like our original perfect integrality gap instance. This implies that our new distributions $\{D'_a\}$ satisfy the required conditions.
\item If $x_{ij} \neq x_i$ for some $i,j$ then exactly half of the linear forms $l_3(XP)$ will be positive and exactly half of the linear forms $l_3(PX)$ will be positive.
\end{enumerate}
Putting these observations together, we have a perfect integrality gap instance $\{f'_i\}$ where each $f'_i$ has the form $sign(l_3)$, as needed.
\end{proof}
\section{Approximation algorithms for predicates}\label{generalapproximationsection}
In this section, we discuss the kinds of approximation algorithms/rounding schemes we must analyze. We note that the ideas here are similar to the ideas of the paper ``Proving Weak Approximability without Algorithms'' by Syed and Tulsiani \cite{ST13}. Indeed, the functions $\{f_a\}$ which we describe below play a central role in their analysis as well. That said, while Syed and Tulsiani use these functions to show that predicates are approximable without actually finding an approximation algorithm, we use these functions to directly show that predicates are approximable by giving an approximation algorithm.

To evaluate the performance of a rounding scheme, we consider the Fourier decomposition $C(x) = \sum_{I \subseteq [1,n]}{\hat{C}_{I}x_I}$ of each constraint $C$ where we take $x_I = \prod_{i\in I}{x_i}$. We will then consider the expected value of each monomial $x_I$ after rounding. To do better than a random assignment, it is sufficient to find a rounding algorithm together with a $\delta > 0$ such that given any point $p \in {KTW}_C$, after rounding we have that 
\[
E[C(x)] = \sum_{I \subseteq [1,n]}{\hat{C}_{I}E[x_I]} > E_{x \in \{-1,+1\}^n}[C(x)] + \delta
\]

We now consider what freedom we have in choosing the expected values of the monomials $\{x_I\}$. As discussed in the previous section, the standard SDP will provide biases $b_i = E[x_i]$ and $b_{ij} = b_{ji} = E[{x_i}{x_j}]$ for each variable and pair of variables. Using these biases, a rounding algorithm $R$ can probabilistically choose values in $\{-1,+1\}$ for the variables $\{x_i\}$. With these choices, we have the following freedom for choosing the expected values of the monomials $\{x_I\}$
\begin{theorem}\label{specifyingapproximationtheorem}
Given continuous functions $f_a: [-1,1]^{a + \binom{a}{2}} \to [-1,1]$ for each $a \in [1,k]$ such that 
\begin{enumerate}
\item For all permutations $\sigma \in S_a$, 
\begin{align*}
&f_a(\{b_{i_{\sigma(j)}}:j \in [1,a]\} \cup \{b_{i_{\sigma(j_1)}i_{\sigma(j_2)}}: j_1,j_2 \in [1,a], j_1 < j_2\}) = \\
&f_a(\{b_{i_j}:j \in [1,a]\} \cup \{b_{i_{j_1}i_{j_2}}: j_1,j_2 \in [1,a], j_1 < j_2\})
\end{align*}
\item For all signs $\{s_{i_j}:j \in [1,a]\} \in \{-1,1\}^a$, 
\begin{align*}
&f_a(\{s_{i_j}b_{i_j}:j \in [1,a]\} \cup \{s_{i_{j_1}}s_{i_{j_2}}b_{i_{j_1}i_{j_2}}: j_1,j_2 \in [1,a], j_1 < j_2\}) = \\
&\left(\prod_{j=1}^{a}{s_{i_j}}\right)f_a(\{b_{i_{j}}:j \in [1,a]\} \cup \{b_{i_{j_1}i_{j_2}}: j_1,j_2 \in [1,a], j_1 < j_2\})
\end{align*}
\end{enumerate}
there exists a sequence of rounding schemes $\{R_q\}$ and real coefficients $\{c_q\}$ such that for all subsets $I$ of size at most $k$, 
\[
\sum_{q}{{c_q}E_{R_q}[v_I]} = f_{|I|}(\{b_{i_j}:j \in [1,k]\} \cup \{b_{i_{j_1}i_{j_2}}: j_1,j_2 \in [1,k], j_1 < j_2\})
\]
Moreover, we can take this sum to be globally convergent.
\end{theorem}
\begin{example}
We can take $f_1$ to be any odd continuous function on one variable. 
\end{example}
\begin{example}
We can take $f_2$ to be any continuous function such that 
\[
\forall b_i,b_j,b_{ij}, f_2(b_i,b_j,b_{ij}) = -f_2(-b_i,b_j,-b_{ij}) = f_2(b_j,b_i,b_{ij})
\]
\end{example}
\begin{example}
We can take $f_3$ to be any contiuous function such that $\forall b_i,b_j,b_k,b_{ij},b_{ik},b_{jk}$, 
\begin{enumerate}
\item $f_3(b_i,b_j,b_k,b_{ij},b_{ik},b_{jk}) = f_3(b_j,b_i,b_k,b_{ij},b_{jk},b_{ik}) = f_3(b_k,b_j,b_i,b_{jk},b_{ik},b_{ij})$
\item $f_3(b_i,b_j,b_k,b_{ij},b_{ik},b_{jk}) = f_3(-b_i,b_j,b_k,-b_{ij},-b_{ik},b_{jk})$
\end{enumerate}
\end{example}
\begin{proof}[Proof sketch of Theorem \ref{specifyingapproximationtheorem}]
Since the KTW polytope is compact, any continuous function on the KTW polytope can be approximated by a polynomial. Thus, it is sufficient to show how we can approximately obtain monomials.
\begin{lemma}
Let $p(\{b_i: i \in [1,m]\},\{b_{ij}: i < j, i,j \in [1,m]\})$ be a monomial. For sufficiently small $\alpha > 0$, given a partition $(V_1,\dots,V_m)$ of the indices $[1,n]$, there is a linear combination of rounding schemes $R$ such that 
\begin{enumerate}
\item Given $(i_1,\dots,i_m)$ such that $i_j \in V_j$, 
\[
\mathbb{E}_{R}\left[\prod_{j=1}^{m}{x_{i_j}}\right] = \alpha^{deg(p)}p(\{b_{i_j}: j \in [1,m]\},\{b_{i_{j_1}i_{j_2}}: j_1 < j_2: j_1,j_2 \in [1,m]\}) \pm O(\alpha^{deg(p)+1})
\]
\item Given a subset of indices $I \subseteq [1,n]$ such that $I$ does not contain exactly one index from each $V_j$, $\left|\sum_{i}{{c_i}E_{R_i}[\prod_{i \in I}{x_{i}}]}\right|$ is $O(\alpha^{deg(p) + 1})$
\end{enumerate}
\end{lemma}
\begin{proof}
\begin{definition}
We define $\chi_{B,\alpha,V}$ to be the operator which does the following:
\begin{enumerate}
\item For all $i \in V$, the operator multiplies $x_i$ by $1$ with probability $\frac{1 + {\alpha}b_{i}}{2}$ and multiplies $x_i$ by $-1$ with probability $\frac{1 + {\alpha}b_{i}}{2}$
\item For all $i \notin V$, the operator keeps $x_i$ as is.
\end{enumerate}
\end{definition}
\begin{lemma}
For any linear combination $R$ of rounding schemes and $I \subseteq [1,n]$, 
\begin{enumerate}
\item If $I \cap V = \emptyset$ then $\mathbb{E}_{\chi_{B,\alpha,V}R}[x_I] = \mathbb{E}_{R}[x_I]$
\item If $I \cap V = \{i\}$ then $\mathbb{E}_{\chi_{B,\alpha,V}R}[x_I] = {\alpha}{b_{i}}\mathbb{E}_{R}[x_I] \pm O(\alpha^2)$
\item If $|I \cap V| \geq 2$ then $\left|\mathbb{E}_{\chi_{B,\alpha,V}R}[x_I]\right|$ is $O({\alpha^2}\left|\mathbb{E}_{R}[x_I]\right|)$
\end{enumerate}
\end{lemma}
\begin{definition}
We define $\chi_{B,\alpha,V,V'}$ to be the operator which does the following:
\begin{enumerate}
\item Finds vectors $\{u_0,u_1,\dots,u_{n}\} \in \mathbb{R}^n$ such that $(\alpha{B} + (1-\alpha)Id)_{ij} = u_i \cdot u_j$.
\item Chooses a random unit vector $w \in \mathbb{R}^n$
\item For all $i \in [1,n]$, multiplies $x_i$ by $sign(w \cdot u_i)$ if $i \in V \cup V'$ and keeps $x_i$ as is if $i \notin V \cup V'$
\end{enumerate}
\end{definition}
\begin{remark}
Although this operator really acts on the single subset $V \cup V'$, we write both $V$ and $V'$ because we are focusing on the case where $I$ has one index in $V$ and one index in $V'$.
\end{remark}
\begin{lemma}
For any linear combination $R$ of rounding schemes and $I \subseteq [1,n]$, 
\begin{enumerate}
\item If $I \cap (V \cup V') = \emptyset$ then $\mathbb{E}_{\chi_{B,\alpha,V,V'}R}[x_I] = \mathbb{E}_{R}[x_I]$
\item If $I \cap (V \cup V') = \{i\}$ then $\mathbb{E}_{\chi_{B,\alpha,V,V'}R}[x_I] = 0$
\item If $I \cap (V \cup V') = \{i,j\}$ then $\mathbb{E}_{\chi_{B,\alpha,V,V'}R}[x_I] =  \frac{2}{\pi}{\alpha}{b_{ij}}\mathbb{E}_{R}[x_I] \pm O(\alpha^2)$
\item If $|I \cap (V \cup V')| \geq 3$ then $\left|\mathbb{E}_{\chi_{B,\alpha,V,V'}R}[x_I]\right|$ is $O({\alpha^2}\left|\mathbb{E}_{R}[x_I]\right|)$
\end{enumerate}
\end{lemma}
\begin{proof}
The first and second parts are trivial. For the third part, the angle between $u_i$ and $u_j$ will be $\approx \frac{\pi}{2} - {\alpha}b_{ij}$ so the probability that $sign((w \cdot u_i)(w \cdot u_j)) = -1$ is $\approx \frac{\frac{\pi}{2} - {\alpha}b_{ij}}{\pi}$ so $E[sign((w \cdot u_i)(w \cdot u_j))] \approx \frac{2}{\pi}{\alpha}b_{ij}$. For the fourth part, since all of the vectors $\{u_i\}$ are nearly orthogonal to each other, for any subset $I$ of size at least $3$, $\left|E\left[sign\left(\prod_{i \in I}{(w \cdot u_i)}\right)\right]\right|$ is $O(\alpha^2)$
\end{proof}
\begin{definition}
Given distinct subsets of vertices $(V_1,\dots,V_j)$, define $P_{V_1,\dots,V_j}$ to be the operator which does the following:
\begin{enumerate}
\item Select each $(y_1,\dots,y_j) \in \{-1,+1\}^{d}$ with weight $2^{-j}\prod_{i=1}^{j}{y_i}$.
\item For each $x_i$, if $i \in V_a$ for some $a \in [1,j]$ then multiply $x_i$ by $y_a$. Otherwise leave $x_i$ as is.
\end{enumerate}
\end{definition}
\begin{proposition}
For any linear combination $R$ of rounding schemes and $I \subseteq [1,n]$, if $I$ has an odd number of elements in each of the subsets $V_1,\dots,V_j$ then $\mathbb{E}_{P_{V_1,\dots,V_j}R}[x_I] = \mathbb{E}_{R}[x_I]$. Otherwise, $\mathbb{E}_{P_{V_1,\dots,V_j}R}[x_I] = 0$.
\end{proposition}
To obtain our final linear combination $R_{final}$ of rounding schemes, we start with the trivial rounding scheme which sets each $x_i$ to $1$ and then apply the following operators
\begin{enumerate}
\item For each term $b_i$ in $p$, we apply the operator $\chi_{B,\alpha,V_i}$
\item For each term $b_{ij}$ in $p$, we apply the operator $\chi_{B,\alpha,V_i,V_j}$
\item We apply the operator $P_(V_1,\dots,V_m)$
\end{enumerate}
Now consider $R_{final}[x_{I}]$. Because of the operator $P(V_1,\dots,V_m)$, $R_{final}[x_{I}] = 0$ unless $I$ contains an odd number of elements in each of the subsets $V_1,\dots,V_m$. From the above lemmas, the dominant terms will be the ones where $I$ contains exactly one element in each of the subsets $V_1,\dots,V_m$, which gives a constant times the desired monomial.
\end{proof}
\end{proof}
\begin{remark}
This theorem is essentially the dual of the KTW criterion.
\end{remark}
\begin{remark}
If our predicate is odd, whenever we have a negative coefficient we can instead flip all of the signs of variables. Thus, any odd predicates (including balanced LTFs) which can be weakly approximated can also be approximated.
\end{remark}
\section{A simpler approximation algorithm for monarchy}\label{monarchysection}
In this section, we give an approximation algorithm for monarchy which is simpler than the original approximation algorithm shown by Austrin, Benabbas, and Magen \cite{ABM10}. We note that the ideas here are similar to the ideas used by Syed and Tulsiani \cite{ST13} to give a proof that monarchy is approximable without giving an approximation algorithm.
\begin{definition}
The monarchy predicate on $k$ variables $x_1,\dots,x_k$ is 
\[
f(x) = sign\left((k-2)x_1 + \sum_{i=2}^{k}{x_i}\right)
\]
We call $x_1$ the president and $x_2,\dots,x_k$ citizens 
\end{definition}
\begin{definition}
We denote the Fourier coefficient of a set of $a$ citizens by $\hat{f}_{aC}$ and we denote the Fourier coefficient of the president together with $a$ citizens by $\hat{f}_{P + aC}$
\end{definition}
\begin{theorem}
The following rounding scheme does better than $\frac{1}{2}$ for the monarchy predicate with $k$ variables.
\begin{enumerate}
\item If $x_i$ has bias $b_i$, after rounding $x_i$ will have bias ${\epsilon}b_i$
\item If $x_{i_1},x_{i_2},x_{i_3}$ have biases $b_{i_1},b_{i_2},b_{i_3}$ then after rounding  $x_{i_1}x_{i_2}x_{i_3}$ will have bias 
\[
C{\epsilon}sign(b_{i_1}b_{i_2}b_{i_3})\min\{|b_{i_1}|,|b_{i_2}|,|b_{i_3}|\}
\] where we take $C = \frac{\hat{f}_{P} - (k-2)\hat{f}_{C}}{\hat{f}_{3C}\binom{k-1}{3} - \hat{f}_{P+2C}\binom{k-1}{2}}$
\end{enumerate}
\end{theorem}
\begin{proof}
Our predicate is $l = (k-2)x_1 + \sum_{i=2}^{k}{x_i}$. Let $\alpha = b_1$ and let $\beta = \sum_{i=2}^{k}{b_i}$. 

We have that $(k-2)\alpha + \beta \geq 1$, but the contribution from the degree 1 terms is 
\[
\epsilon\left(\hat{f}_{P}\alpha + \hat{f}_{C}\beta\right)
\]
where $\hat{f}_{C}$ is exponentially smaller than $\hat{f}_{P}$. We compensate for this using the degree 3 terms.
\begin{lemma} \ 
\begin{enumerate}
\item $\sum_{i_1,i_2,i_3 \in [2,k]:i_1 < i_2 < i_3}{sign(b_{i_1}b_{i_2}b_{i_3})\min\{|b_{i_1}|,|b_{i_2}|,|b_{i_3}|\}} \geq -\binom{k-1}{3}\alpha$
\item $\sum_{i_1,i_2 \in [2,k]:i_1 < i_2}{sign(b_{1}b_{i_1}b_{i_2})\min\{|b_{1}|,|b_{i_1}|,|b_{i_2}|\}} \leq \binom{k-1}{2}\alpha$
\end{enumerate}
\end{lemma}
\begin{proof}
\begin{proposition}
$\forall i \in [2,k], b_i \geq -\alpha$
\end{proposition}
\begin{proof}
Observe that the only satisfying assignment which has $x_1 = -1$ is $x_1 = -1, x_i = 1$ for all $i \in [2,k]$. Thus, for all $i \in [2,k]$ we always have that $x_i \geq -x_1$ and the result follows.
\end{proof}
With this proposition in mind, we have two cases. Either $\alpha \leq 0$ or $\alpha > 0$. If $\alpha \leq 0$ then
\begin{enumerate}
\item For all $i_1,i_2,i_3 \in [2,k]$, $\min\{|b_{i_1}|,|b_{i_2}|,|b_{i_3}|\} \geq -\alpha$ so 
\[
\sum_{i_1,i_2,i_3 \in [2,k]:i_1 < i_2 < i_3}{sign(b_{i_1}b_{i_2}b_{i_3})\min\{|b_{i_1}|,|b_{i_2}|,|b_{i_3}|\}} \geq -\binom{k-1}{3}\alpha
\]
\item For all $i_1,i_2 \in [2,k]$, $\min\{|b_{1}|,|b_{i_1}|,|b_{i_2}|\} = -\alpha$ so 
\[
\sum_{i_1,i_2 \in [2,k]:i_1 < i_2}{sign(b_{1}b_{i_1}b_{i_2})\min\{|b_{1}|,|b_{i_1}|,|b_{i_2}|\}} = \binom{k-1}{2}\alpha
\]
\end{enumerate}
Similarly, if $\alpha > 0$ then 
\begin{enumerate}
\item For all $i_1,i_2,i_3 \in [2,k]$, if $b_{i_1}b_{i_2}b_{i_3} < 0$ then $\min\{|b_{i_1}|,|b_{i_2}|,|b_{i_3}|\} \leq \alpha$ so 
\[
\sum_{i_1,i_2,i_3 \in [2,k]:i_1 < i_2 < i_3}{sign(b_{i_1}b_{i_2}b_{i_3})\min\{|b_{i_1}|,|b_{i_2}|,|b_{i_3}|\}} \geq -\binom{k-1}{3}\alpha
\]
\item For all $i_1,i_2 \in [2,k]$, $\min\{|b_{1}|,|b_{i_1}|,|b_{i_2}|\} \leq \alpha$ so 
\[
\sum_{i_1,i_2 \in [2,k]:i_1 < i_2}{sign(b_{1}b_{i_1}b_{i_2})\min\{|b_{1}|,b_{i_1},b_{i_2}|\}} \leq \binom{k-1}{2}\alpha
\]
\end{enumerate}
\end{proof}
Thus, since $\hat{f}_{3C} > 0$ and $\hat{f}_{P+2C} < 0$, the contribution from the degree 3 terms is at least 
\[
C{\epsilon}\left(\hat{f}_{P+2C}\binom{k-1}{2} - \hat{f}_{3C}\binom{k-1}{3}\right)\alpha \geq \epsilon((k-2)\hat{f}_{C} - \hat{f}_{P})\alpha
\]
Adding the contributions together, we obtain
\[
\epsilon\left(\hat{f}_{P}\alpha + \hat{f}_{C}\beta\right) + \epsilon((k-2)\hat{f}_{C} - \hat{f}_{P})\alpha = {\epsilon}((k-2)\alpha + \beta)\hat{f}_{C} \geq {\epsilon}\hat{f}_{C} > 0
\]
\end{proof}
\section{Approximation algorithm for almost monarchy}\label{almostmonarchysection}
In this section, we show that the almost monarchy predicate is approximable for sufficiently large $k$.
\begin{definition}
The almost monarchy predicate on $k$ variables $x_1,\dots,x_k$ is 
\[
f(x) = sign\left((k-4)x_1 + \sum_{i=2}^{k}{x_i}\right)
\]
We call $x_1$ the president and $x_2,\dots,x_k$ citizens 
\end{definition}
\begin{remark}
For the almost monarchy predicate, the president gets his/her way as long as he/she has at least two supporters among the citizens.
\end{remark}
\begin{definition}
We take $\alpha = b_1$ and we take $\beta = \sum_{i=2}^{k}{b_i}$
\end{definition}
\begin{theorem}
For sufficiently large $k$, the almost monarchy predicate on $k$ variables is approximable.
\end{theorem}
To prove this theorem, we use the following rounding scheme.
\begin{enumerate}
\item After rounding, $x_i$ will have bias ${\epsilon}b_i$
\item After rounding, $x_{i_1}x_{i_2}x_{i_3}$ will have bias 
\[
3\frac{C{\epsilon}}{E}(b_{i_1}b_{{i_2}{i_3}} + b_{i_2}b_{{i_1}{i_3}} + b_{i_3}b_{{i_1}{i_2}})
\] 
\item After rounding, $x_{i_1}x_{i_2}x_{i_3}x_{i_4}x_{i_5}$ will have bias 
\[
-6\frac{C{\epsilon}}{E^2}(b_{i_1}b_{{i_2}{i_3}}b_{{i_4}{i_5}} + \text{ symmetric terms})
\] 
\item After rounding, $x_{i_1}x_{i_2}x_{i_3}x_{i_4}x_{i_5}x_{i_6}x_{i_7}$ will have bias 
\[
6\frac{C{\epsilon}}{E^3}(b_{i_1}b_{{i_2}{i_3}}b_{{i_4}{i_5}}b_{{i_6}{i_7}} + \text{ symmetric terms})
\] 
\end{enumerate}
where we take $E = \frac{k^2 - 9k + 18}{2}$ and we take 
\[
C = \frac{2^{k-2}\left(\hat{f}_{P} - (k-4)\hat{f}_{C}\right)}{(k-2)(k-3)}
\]
We now give a sketch for why this rounding scheme does better than random (when the SDP thinks almost all constraints are satisfiable). We will then give a full proof.
\begin{proof}[Proof sketch]
Consider the expression
\[
S_{\{(i_1,i_2)\}} = \sum_{i_1,i_2 \in [2,k]: i_1 < i_2}{b_{{i_1}{i_2}}}
\]
Roughly speaking, the least favorable accepting assignments are the ones where the monarch votes YES and all but two of the citizens vote NO or the monarch votes NO and all but one of the citizens vote YES. In these cases, $S_{\{(i_1,i_2)\}} \approx E$. In these cases, writing $S_{\{(i_1,i_2)\}} = E(1 + \Delta)$, we expect $\Delta$ to be small. 

Now observe that we have the following approximations
\begin{enumerate}
\item $\sum_{i_1 < i_2 < i_3 \in [2,k]}{\left(b_{i_1}b_{{i_2}{i_3}} + b_{i_2}b_{{i_1}{i_3}} + b_{i_3}b_{{i_1}{i_2}}\right)} \approx {\beta}S_{\{(i_1,i_2)\}} = {\beta}E(1+\Delta)$
\item $2\sum_{i_1 < i_2 < i_3 \in [2,k]}{\left(b_{i_1}b_{{i_2}{i_3}}b_{{i_4}{i_5}} + \text{ symmetric terms}\right)} \approx {\beta}(S_{\{(i_1,i_2)\}})^2 = {\beta}{E^2}(1+\Delta)^2$
\item $6\sum_{i_1 < i_2 < i_3 \in [2,k]}{\left(b_{i_1}b_{{i_2}{i_3}}b_{{i_4}{i_5}}b_{{i_6}{i_7}} + \text{ symmetric terms}\right)} \approx {\beta}(S_{\{(i_1,i_2)\}})^3 = {\beta}{E^3}(1+\Delta)^3$
\end{enumerate}
When $|\Delta|$ is small, all of these expressions are close to a multiple of $\beta$ and could potentially be used to counteract the excess $\alpha$ from the degree 1 terms. However, by taking a linear combination of these expressions, we can obtain an even better approximation to $\beta$. Observe that
\[
3(1 + \Delta) - 3(1 + \Delta)^2 + (1 + \Delta)^3 = ((1 + \Delta) - 1)^3 + 1 = {\Delta}^3 + 1
\]
Thus, the contribution from the degree 3,5, and 7 terms is approximately ${\epsilon}C(1 + {\Delta}^3){\beta}$. When $|\Delta|$ is small, this is close enough to ${\epsilon}C{\beta}$ to counteract the extra $\alpha$ in the degree 1 terms, giving us a positive value. When $|\Delta|$ is large, the degree 1 terms are very favorable for us which gives us a positive value.
\end{proof}
\begin{proof}[Full proof]
Before giving the full proof, we need some preliminaries.

We need the following Fourier coefficients
\begin{lemma} \ 
\begin{enumerate}
\item For all odd $a$, $\hat{f}_{aC} = \frac{k - 2a}{2^{k-2}}$
\item Whenever $a$ is even and $a \geq 2$, $\hat{f}_{P+aC} = \frac{2a-k}{2^{k-2}}$
\end{enumerate}
\end{lemma}
For a proof, see Appendix \ref{Fouriercoefficientappendix}
\begin{corollary}
We have the following Fourier coefficients:
\begin{enumerate}
\item $\hat{f}_{C} = \frac{k-2}{2^{k-2}}$
\item $\hat{f}_{3C} = \frac{k-6}{2^{k-2}}$
\item $\hat{f}_{5C} = \frac{k-10}{2^{k-2}}$
\item $\hat{f}_{7C} = \frac{k-14}{2^{k-2}}$
\item $\hat{f}_{P} = 1 - \frac{k}{2^{k-2}}$
\item $\hat{f}_{P+2C} = \frac{4-k}{2^{k-2}}$
\item $\hat{f}_{P+4C} = \frac{8-k}{2^{k-2}}$
\item $\hat{f}_{P+6C} = \frac{12-k}{2^{k-2}}$
\end{enumerate}
\end{corollary}

We use the following notation for sums of products of biases. 
\begin{definition}
Given a set of indices $E_1$ and a set of edges $E_2$, define 
\[
B_{E_1 \cup E_2} = \prod_{i \in E_1}{b_i} \cdot \prod_{(i,j) \in E_2}{b_{ij}}
\]
\end{definition}
\begin{definition}
Given a hypergraph $H$ such that the vertices of $H$ are labeled with either $\alpha$ or an unspecified index $i_j$ and the edges of $H$ have arity $1$ or $2$, we define
\[
S_{H} = \sum_{E_1,E_2: \exists \sigma: V(H) \to [1,k]: \sigma \text{ is injective }, \atop \sigma(\alpha) = 1, \sigma(E(H)) = E_1 \cup E_2}{B_{E_1 \cup E_2}}
\]
We write down each such $H$ by writing down the edges in each connected component of $H$ within $\{\}$ brackets.
\end{definition}

We now have the following approximate equalities, all of which can be proved with inclusion/exclusion. For the exact equalities and images corresponding to the calculations, see the appendix.
\begin{lemma} \ 
\begin{enumerate}
\item $S_{\{i_1\},\{(i_2,i_3)\}} = {\beta}S_{\{(i_1,i_2)\}} - S_{\{i_1,(i_1,i_2)\}}$
\item $S_{\{\alpha\},\{(i_2,i_3)\}} = {\alpha}S_{\{(i_1,i_2)\}}$
\item $S_{\{i_1\},\{\alpha,i_2\}} = {\beta}S_{\{\alpha,i_1\}} \pm{O(k)}$
\item $2S_{\{i_1\},\{(i_2,i_3)\},\{(i_4,i_5)\}} = {\beta}(S_{\{(i_1,i_2)\}})^2 - 2S_{\{i_1,(i_1,i_2)\}}S_{\{(i_1,i_2)\}} - 2{\beta}S_{\{(i_1,i_2),(i_1,i_3)\}} \pm{O(k^3)}$
\item $2S_{\{\alpha\},\{(i_2,i_3)\},\{(i_4,i_5)\}} = {\alpha}(S_{\{(i_1,i_2)\}})^2 \pm{O(k^3)}$
\item $S_{\{i_1\},\{(\alpha,i_2)\},\{(i_3,i_4)\}} = {\beta}S_{\{(\alpha,i_1)\}}S_{\{(i_1,i_2)\}} \pm O(k^3)$
\item $6S_{\{i_1\},\{(i_2,i_3)\},\{(i_4,i_5)\},\{i_6,i_7\}} = {\beta}(S_{\{(i_1,i_2)\}})^3 - 3S_{\{i_1,(i_1,i_2)\}}(S_{\{(i_1,i_2)\}})^2 - 6{\beta}S_{\{(i_1,i_2),(i_1,i_3)\}}S_{\{(i_1,i_2)\}} \pm{O(k^5)}$
\item $6S_{\{\alpha\},\{(i_2,i_3)\},\{(i_4,i_5)\},\{i_6,i_7\}} = {\alpha}(S_{\{(i_1,i_2)\}})^3 \pm{O(k^5)}$
\item $2S_{\{i_1\},\{(\alpha,i_2)\},\{(i_3,i_4)\},\{i_5,i_6\}} = {\beta}{S_{\{(\alpha,i_1)\}}}(S_{\{(i_1,i_2)\}})^2 \pm{O(k^5)}$
\end{enumerate}
\end{lemma}
\begin{definition}
We write $S_{\{i_1,i_2\}} = E(1+\Delta)$
\end{definition}
For our approximation algorithm, we obtain the following contributions for the degree $3$, $5$, and $7$ terms divided by $C{\epsilon}2^{2-k}$ (ignoring terms of size $O(1)$):
\begin{enumerate}
\item $(k-6)\left(3{\beta}(1 + \Delta) - \frac{3}{E}{S_{\{i_1,(i_1,i_2)\}}}\right)$
\item $-3(k-4){\alpha}(1 + \Delta)$
\item $-3(k-4){\beta}S_{\{\alpha,i_1\}}$
\item $(k-10)\left(-3{\beta}(1 + \Delta)^2 + \frac{6}{E}(1 + \Delta)S_{\{i_1,(i_1,i_2)\}}\right)$
\item $-3(k-8){\alpha}(1 + \Delta)^2$
\item $-\frac{6(k-8)}{E}{\beta}(1 + \Delta)S_{\{(\alpha,i_1)\}}$
\item $(k-14)\left({\beta}(1 + \Delta)^3 - \frac{3}{E}(1 + \Delta)^2{S_{\{i_1,(i_1,i_2)\}}}\right)$
\item $-{\alpha}(k-12)(1 + \Delta)^3$
\item $-{\beta}(k-12)\frac{3}{E}(1 + \Delta)^2{S_{\{(\alpha,i_1)\}}}$
\end{enumerate}
Adding up these terms, we obtain
\begin{align*}
&(k-2)\left(3(1+\Delta) - 3(1 + \Delta)^2 + (1 + \Delta)^3\right)\beta \\
&-4\left(3(1+\Delta) - 6(1 + \Delta)^2 + 3(1 + \Delta)^3\right)\beta \\
&+ (k-6)\frac{3}{E}\left(1 - 2(1 + \Delta) + (1 + \Delta)^2\right)S_{\{i_1,(i_1,i_2)\}} \\
&- 4\frac{3}{E}\left(-2(1 + \Delta) + 2(1 + \Delta)^2\right)S_{\{i_1,(i_1,i_2)\}} \\
&-k\left(3(1+\Delta) - 3(1 + \Delta)^2 + (1 + \Delta)^3\right)\alpha \\
&+4\left(3(1+\Delta) - 6(1 + \Delta)^2 + 3(1 + \Delta)^3\right)\alpha \\
&- (k-4)\frac{1}{E}\left(1 - 2(1 + \Delta) + (1 + \Delta)^2\right){\beta}{S_{\{(\alpha,i_1)\}}}\\
&+ 4\frac{1}{E}\left(-2(1 + \Delta) + 2(1 + \Delta)^2\right){\beta}{S_{\{(\alpha,i_1)\}}}\\
&= (k-2)\beta + (k-2){\Delta}^3{\beta} - 4(3{\Delta}^2 + 3{\Delta}^3)\beta \\
&- k\alpha - k{\Delta}^3{\alpha} + 4(3{\Delta}^2 + 3{\Delta}^3){\alpha} \\
&+ \frac{3(k-6)}{E}{\Delta}^2{S_{\{i_1,(i_1,i_2)\}}} - \frac{12}{E}(2\Delta + 2{\Delta}^2){S_{\{i_1,(i_1,i_2)\}}} \\
&- \frac{(k-4)}{E}{\Delta}^2{\beta}{S_{\{(\alpha,i_1)\}}} + \frac{4}{E}(2\Delta + 2{\Delta}^2){\beta}{S_{\{(\alpha,i_1)\}}}
\end{align*}
Ignoring negligible terms, this is
\[
(k-2)(1 + {\Delta}^3)(\beta - \alpha) -4(3{\Delta}^2 + 3{\Delta}^3)\beta + \frac{3k}{E}{\Delta}^2{S_{\{i_1,(i_1,i_2)\}}} - \frac{k}{E}{\Delta}^2{\beta}{S_{\{\alpha,i_1\}}}
\]
We focus on the term $(k-2)(1 + {\Delta}^3)(\beta - \alpha)$ as this term has magnitude $O(k^2)$ while the other terms have magnitude $O(k)$
\begin{lemma}
For sufficiently large $k$, we always have that $(k-4)\alpha + \beta \geq \frac{1}{3} + \frac{(k-6)|\Delta|}{3}$
\end{lemma}
\begin{proof}
To prove this statement, it is sufficient to check that this statement holds for each individual satisfying assignment as $\alpha$ and $\beta$ are linear functions and $|\Delta|$ is convex.

If $x_1 = -1$ then there are either $0$ or $1$ $-1$s in $\{x_2,\dots,x_k\}$. If there is exactly one $-1$ in $\{x_2,\dots,x_k\}$ then 
$(k-4)x_1 + \sum_{i=1}^{k}{x_i} = 1$ and $\sum_{i,j \in [2,k]: i < j}{{x_i}{x_j}} = \binom{k-2}{2} - (k-2) = \frac{(k-2)(k-5)}{2}$. If there are no $-1$s in $\{x_2,\dots,x_k\}$ then $(k-4)x_1 + \sum_{i=1}^{k}{x_i} = 3$ and $\sum_{i,j \in [2,k]: i < j}{{x_i}{x_j}} = \binom{k-1}{2} = \frac{(k-1)(k-2)}{2}$.

If $x_1 = 1$ then there can be at most $(k-3)$ $-1$s in $\{x_2,\dots,x_k\}$. If there is exactly $(k-3)$ $-1$s in $\{x_2,\dots,x_k\}$ then $(k-4)x_1 + \sum_{i=1}^{k}{x_i} = 1$ and 
\[
\sum_{i,j \in [2,k]: i < j}{{x_i}{x_j}} = \binom{k-3}{2} - 2(k-3) + 1 = \frac{(k-3)(k-4) - 4(k-3) + 2}{2} = \frac{k^2 - 11k + 26}{2}
\]
Since $E = \frac{k^2 - 9k + 18}{2}$, observe that the deviation of $\sum_{i,j \in [2,k]: i < j}{{x_i}{x_j}}$ from $E$ is at most $\frac{3}{2}(k-4)$ times $(k-4)x_1 + \sum_{i=1}^{k}{x_i} - \frac{1}{3}$. Thus, $|{\Delta}|E \leq \frac{3}{2}(k-4)\left((k-4)x_1 + \sum_{i=1}^{k}{x_i} - \frac{1}{3}\right)$ which implies that
\[
(k-4)x_1 + \sum_{i=1}^{k}{x_i} \geq \frac{1}{3} + \frac{2|{\Delta}|E}{3(k-4)} \geq \frac{1}{3} + \frac{(k-6)|\Delta|}{3}
\]
\end{proof}
We now consider what happens when we add $(k-2)(k-3)\alpha$ to the contribution from the degree 3,5, and 7 terms, which ignoring negligible terms was
\[
(k-2)(1 + {\Delta}^3)(\beta - \alpha) -4(3{\Delta}^2 + 3{\Delta}^3)\beta + \frac{3k}{E}{\Delta}^2{S_{\{i_1,(i_1,i_2)\}}} - \frac{k}{E}{\Delta}^2{\beta}{S_{\{\alpha,i_1\}}}
\] 
We have the following cases
\begin{enumerate}
\item If $\Delta > -.55$ then ${\Delta}^2 < .3025 < \frac{1}{3} - \frac{3}{100}$ so for sufficiently large $k$, 
\begin{align*}
(k-2)(k-3)\alpha + (k-2)(1 + {\Delta}^3)(\beta - \alpha) &= (k-2)((k-4)\alpha + \beta) + (k-2){\Delta}^3(\beta - \alpha) \\
&\geq \frac{(k-2)}{3} + \frac{(k-2)(k-6)|\Delta|}{3} + (k-2){\Delta}^3(\beta - \alpha) \\
&\geq \frac{(k-2)}{3} + \frac{3(k-2)(k-6)|\Delta|}{100}
\end{align*}
For sufficiently large $k$ we have that
\[
\frac{3(k-2)(k-6)|\Delta|}{100} > \left|-4(3{\Delta}^2 + 3{\Delta}^3)\beta + \frac{3k}{E}{\Delta}^2{S_{\{i_1,(i_1,i_2)\}}} - \frac{k}{E}{\Delta}^2{\beta}{S_{\{\alpha,i_1\}}}\right|
\]
so the remaining terms are dominated and we are guaranteed to have a positive value.
\item If $\Delta \leq -.55$ and $k$ is sufficiently large then we must have $\alpha > 0$. If so, we can instead express the largest terms as a positive linear combination of $(1 + {\Delta}^3)((k-4)\alpha + \beta) $ and $\alpha$. In particular,
\[
(k-2)(k-3)\alpha + (k-2)(1 + {\Delta}^3)(\beta - \alpha) = (k-2)(1 + {\Delta}^3)((k-4)\alpha + \beta) - (k-2)(k-3){\Delta}^3{\alpha}
\]
If $\Delta \geq -1$ then both of these terms will be non-negative and at least one term will be $\Omega(k^2)$. The minimum possible value of $\Delta$ is $-1 - O(\frac{1}{k})$ so if $\Delta < -1$ then the second term is $\Omega(k^2)$ and it dominates the first term. Either way, for sufficiently large $k$, these terms will dominate the remaining terms and we are again guaranteed to have a positive value.
\end{enumerate}
\end{proof}
\section{Further Work}
There are several possible questions for further research. A few of these questions are as follows
\begin{enumerate}
\item Our work gives renewed impetus to the following question: Can we prove sum of squares lower bounds for approximating any CSP which is unique games hard to approximate? Prior to our work, thanks to the sum of squares bounds of Kothari, Mori, O'Donnell, and Witmer \cite{KMOW17} on random CSPs, we did not have an example of a predicate $P$ which is unique games hard to approximate for which sum of squares lower bounds were unknown. With this work, we now have such a predicate $P$.
\item Can we find a higher degree core, i.e. a core where we in addition to specifying the expected values $c_i,c_{ii},c_{ij}$ of $x_i,x^2_i,x_{ij}$ we also specify the expected values of higher degree monomials?
\begin{remark}
The kind of question this would answer is as follows. What degree Fourier coefficients do we need to look at in order to distinguish between the case when all of our balanced LTFs are satisfiable and at most half of our balanced LTFs are satisfiable?
\end{remark}
\item Can we generalize the techniques we used to find an approximation algorithm for almost monarchy to find approximation algorithms for other balanced LTFs. In particular, can we prove that any presidential type predicate (i.e. a predicate of the form $f(x_1,\dots,x_k) = sign(c(k)x_1 + \sum_{i=2}^{k}{x_i})$) is approximable?
\item Are there any predicates $P$ which are unique games hard to approximate but can be weakly approximated?
\item Are there any predicates $P$ which are unique games hard to weakly approximate such that either any measure $\Lambda$ certifying the KTW criterion for $P$ must have more than one point. Similarly, are there any predicates $P$ which are unique games hard to approximate for which there are no perfect integrality gap instances?
\end{enumerate}
\noindent Acknowledgements: The author would like to thank Per Austrin, Johan H{\aa}stad, and Joseph Swernovsky for helpful conversations. The author would also like to thank Johan H{\aa}stad for helpful comments on the paper. This work was supported by the Knut and Alice Wallenberg Foundation, the European Research Council, and the Swedish Research Council.

\begin{appendix}
\section{Verifying the KTW Criterion with a Perfect Integrality Gap Instance}\label{verifyingKTWappendix}
In this section, we verify that if we have a perfect integrality gap instance then the KTW criterion is satisfied.
\begin{lemma}
If there is a perfect integrality gap instance $\{f_a: a \in [1,m]\}$ of functions of form $P$ then $P$ satisfies the KTW criterion for being unique games hard to weakly approximate.
\end{lemma}
\begin{proof}
To see that the KTW criterion for being unique games hard to weakly approximate is satisfied, if $f_a = P(z_{a1}x_{\phi_a(1)},\dots,z_{ak}x_{\phi_a(k)})$ then take 
\[
p_a := \{z_{ai}b_{\phi_{a}(i)}: i \in [1,k]\} \text{ concatenated with } \{z_{ai}z_{aj}b_{\phi_{a}(i)\phi_{a}(j)}: i,j \in [1,k], i < j\}
\]
where the pairs $i,j$ are in lexicographical order. Now take $\Lambda$ to be 
\[
\Lambda = \sum_{a=1}^{m}{\frac{1}{m}\delta_{p_a}}
\]
where $\delta_{p}$ is the measure with weight one on the point $p$ and weight $0$ elsewhere. For each subset $S = \{s_1,\dots,s_t\} \subseteq [1,k]$ of size $t$, each permutation $\pi: [1,t] \to [1,t]$, and signs $z' = (z'_1,\dots,z'_t) \in \{-1,+1\}^t$, define $\tau_S: [1,t] \to [1,k]$ to be the map where $\tau_S(i) = s_i$ and define 
\begin{align*}
p_{a,S,\pi,z'} &:= \{z'_{i}z_{a\tau_S(i)}b_{\phi_{a}(\tau_S(\pi(i)))}: i \in [1,t]\} \text{ concatenated with } \\
&\{{z'_i}{z'_j}z_{a\tau_S(i)}z_{a\tau_S(j)}b_{\phi_{a}(\tau_S(\pi(i)))\phi_{a}(\tau_S(\pi(j)))}: i,j \in [1,t], i < j\}
\end{align*}
We have that 
\begin{align*}
\Lambda_{P}^{(t)} &= \mathbb{E}_{S \subseteq [1,k]:|S| = t}\mathbb{E}_{\pi:[1,t] \to [1,t]}\mathbb{E}_{z'=(z'_1,\dots,z'_t) \in \{-1,1\}^t}\left[\left(\prod_{i=1}^{t}{z'_i}\right) \cdot \hat{P}_S \cdot \Lambda_{S,\pi,z'}\right] \\
&= \frac{1}{m}\sum_{a=1}^{m}{\mathbb{E}_{S:|S| = t}\mathbb{E}_{\pi:[1,t] \to [1,t]}\mathbb{E}_{z' = (z'_1,\dots,z'_t) \in \{-1,1\}^t}\left[\left(\prod_{i=1}^{t}{z'_i}\right) \cdot \left(\prod_{j \in S}{z_{aj}}\right) \cdot \hat{(f_a)}_{\phi_a(S)} \cdot \delta_{p_{a,S,\pi,z'}}\right]}
\end{align*}
To show that this is $0$, it is sufficient to show that for all subsets $T \subseteq [1,n]$ of size $t$,
\[
\sum_{a \in [1,m], S \subseteq [1,k]:\phi_a(S) = T}{\mathbb{E}_{\pi:[1,t] \to [1,t]}\mathbb{E}_{z' = (z'_1,\dots,z'_t) \in \{-1,1\}^t}\left[\left(\prod_{i=1}^{t}{z'_i}\right) \cdot \left(\prod_{j \in S}{z_{aj}}\right) \cdot \hat{(f_a)}_{T} \cdot \delta_{p_{a,S,\pi,z'}}\right]} = 0
\]
\end{proof}
To see that this expression is zero, take $z'_i = (w_{a,S})_{i}z_{a\tau_S(i)}$, take $\tau_T:[1,t] \to T$ to be the map such that $\tau_{T}(i) = t_i$, and for each permutation $\pi': T \to T$ and $w \in \{-1,+1\}^t$, define 
\[
p_{T,\pi',w} := \{w_{i}b_{\pi'(\tau_T(i))}: i \in [1,t]\} \text{ concatenated with } \{w_{i}w_{j}b_{\pi'(\tau_T(i))\pi'(\tau_T(j))}: i,j \in [1,t], i < j\}
\]
With these definitions, our expression is equal to 
\begin{align*}
&\sum_{a \in [1,m], S \subseteq [1,k]:\phi_a(S) = T}{\mathbb{E}_{\pi':[1,t] \to [1,t]}\mathbb{E}_{w_{a,S} \in \{-1,1\}^t}\left[\left(\prod_{i=1}^{t}{(w_{a,S})_i}\right) \cdot \hat{(f_a)}_{T} \cdot \delta_{p_{T,\pi',w_{a,S}}}\right]} \\
&=\left(\sum_{a \in [1,m], S \subseteq [1,k]:\phi_a(S) = T}{\hat{(f_a)}_{T}}\right)
\mathbb{E}_{\pi':[1,t] \to [1,t]}\mathbb{E}_{w \in \{-1,1\}^t}\left[\left(\prod_{i=1}^{t}{w_i}\right) \cdot \delta_{p_{T,\pi',w}}\right] \\
&= \widehat{\left(\sum_{a=1}^{m}{f_a}\right)}_T
\mathbb{E}_{\pi':[1,t] \to [1,t]}\mathbb{E}_{w \in \{-1,1\}^t}\left[\left(\prod_{i=1}^{t}{w_i}\right) \cdot \delta_{p_{T,\pi',w}}\right] 
\end{align*}
which is $0$ because $|T| = t \geq 1$ and $\sum_{a=1}^{m}{f_a}$ is a constant function.
\section{Finding the core}\label{findingcoreappendix}
In this section, we describe how we found the core given in subsection \ref{coresubsection}.

Some intuition is as follows. By changing coordinates, we can assume that the coefficients $\{c_i\}$ and $\{c_{ij}\}$ are all $0$ and all of the $\{c_{ij}\}$ are equal to each other. We can further assume that one of the linear forms is approximately $x_1 + z$ for some constrant $z$. We want a set of possible solution vectors such that
\begin{enumerate}
\item Exactly half of the possible solution vectors are on either side of the hyperplane $x_1 = -z$  
\item There is a distribution $D_1$ such that 
\begin{enumerate}
\item $D_1$ supported on vectors $x$ in our set of possible solution vectors such that $x_1 > -z$
\item $E_{D_1}[x_i] = 0$, $E_{D_1}[{x_i}{x_j}] = 0$, and $E_{D_1}[{x^2_i}] = c_ii$
\end{enumerate}
\end{enumerate}
In oder to satisfy the first condition, we want $z$ to be positive but as small as possible. However, to satisfy the second condition we need $x_1$ to have high variance and mean zero in the distribution $D_1$. In order for $x_1$ to have high variance, mean zero, and a minimum value which is not much less than $0$, we take the distribution $D_1$ so that $x_1$ is almost always equal to its minimum value but is very large with small probability. This suggests that we take solution vectors of the form $(-a,-a,-b,-b)$ and $(e,0,0,0)$ where $e >> b > z > a$. However, if we do this then the other coordinates of $x$ will not have a high enough variance. 

We can fix this by also taking solution vectors of the form $(-a,-c,-c,d)$ where $d$ is large but not as large as $e$. Thus, it is reasonable to try searching for cores of the following form:
\begin{enumerate}
\item Our vectors are permutations of one of the following vectors
\begin{enumerate}
\item $(-a,-a,-b,-b)$
\item $(-a,-c,-c,d)$
\item $(e,0,0,0)$
\end{enumerate}
\item Our coefficients are 
\begin{enumerate}
\item $\forall i \in [1,4], c_i = 0$
\item $\forall i < j \in [1,4], c_{ii} = c_{jj}$
\item $\forall i < j \in [1,4], c_{ij} = 0$
\end{enumerate}
\item The distribution $D_1$ has the form 
\begin{enumerate}
\item Take each of the vectors $(-a,-a,-b,-b)$, $(-a,-b,-a,-b)$, and $(-a,-a,-b,-b)$ with probability $p_1$.
\item Take each of the vectors $(-a,-c,-c,d)$, $(-a,-c,-c,d)$, and $(-a,-c,-c,d)$ with probability $p_2$.
\item Take the vector $(e,0,0,0)$ with probability $p_3$.
\end{enumerate}
and the other distributions are symmetric to $D_1$
\end{enumerate}
In order to match the coefficients $\{c_i\}$, $\{c_{ii}\}$, $\{c_{ij}\}$ with such a core, we must satisfy the following equations
\begin{enumerate}
\item $E[x_1] = -3a{p_1} - 3a{p_2} + e{p_3} = 0$
\item $E[x_2] = -(2b + a)p_1 + (d - 2c)p_2 = 0$
\item $E[{x_1}{x_2}] = a(2b + a)p_1 - a(d - 2c)p_2 = 0$
\item $E[{x_2}{x_3}] = (b^2 + 2ab)p_1 + (c^2 - 2cd)p_2 = 0$
\item $E[x^2_1] = 3a^2{p_1} + 3a^2{p_2} + e^2{p_3} = E[x^2_2] = (2b^2 + a^2)p_1 + (2c^2 + d^2)p_2$
\item $3p_1 + 3p_2 + p_3 = 1$
\end{enumerate}
Note that the third equation is $-a$ times the second equation so it is redundant. Rearranging the second equation and fourth equations gives
\[
\frac{p_2}{p_1} = \frac{2b + a}{d - 2c} = \frac{b(b+2a)}{c(2d - c)}
\]
As long as we satisfy the equation $\frac{2b + a}{d - 2c} = \frac{b(b+2a)}{c(2d - c)}$, we will have enough degrees of freedom with $p_1,p_2,p_3,e$ to satisfy the remaining equations.

Intuitively, $-a$ should be just barely negative, so let's try $a = 1$. $-c$ should be more negative than $-a$, but not by much, so let's try $c = 2$. With these values, we have that $\frac{2b + 1}{d - 4} = \frac{b^2 + 2b}{4d - 4}$. Rearranging, this implies that $(b^2 + 2b)(d - 4) = (2b + 1)(4d - 4)$. Rearranging this equation we obtain that
\[
(b^2 - 6b - 4)d = 4b^2 - 4
\]
We want that $d - 2c > 0$ and $b > 1$, so let's try $b = 7$. This gives $3d = 192$ so $d = 64$. Thus we have that $a = 1$, $b = 7$, $c = 2$, and $d = 64$.

We now make the following deductions to find $e,p_1,p_2,p_3$:
\begin{enumerate}
\item Looking at the second equation, $15p_1 = 60p_2$ so $p_1 = 4p_2$
\item Looking at the first equation, $e{p_3} = 3(p_1 + p_2) = 15{p_2}$
\item Looking at the fifth equation, $3(p_1 + p_2) + {e^2}{p_3} = 15{p_2} + 15{p_2}e = 99{p_1} + 4104{p_2} = 4500{p_2}$ so $e = \frac{4500-15}{15} = 299$
\item Plugging this into the first equation we obtain that $p_3 = \frac{15}{299}p_2$
\item Using the final equation $3p_1 + 3p_2+ p_3 = 1$ we obtain that 
\[
15 + \frac{15}{299}p_2 = \frac{4500}{299}p_2 = 1
\]
so $p_2 = \frac{299}{4500}$, $p_1 = \frac{1196}{4500}$, and $p_3 = \frac{15}{4500} = \frac{1}{300}$
\end{enumerate}
\section{Fourier coefficient calculations}\label{Fouriercoefficientappendix}
In this section, we compute the Fourier coefficients of the almost monarchy predicate 
\[f(x_1,\dots,x_k) = sign((k-4)x_1 + \sum_{i=2}^{k}{x_i})
\]
\begin{lemma}
For all odd $a$, $\hat{f}_{aC} = \frac{k - 2a}{2^{k-2}}$
\end{lemma}
\begin{proof}
We first choose what happens with the president and the other $k-a-1$ citizens and then consider the resulting function on $a$ citizens. The probabilities are as follows:
\begin{enumerate}
\item With probability $2^{a-k}$, the president votes no but all the remaining citizens vote yes. If so, $f$ is $1$ if and only if there is at most $1$ no in the $a$ citizens. The Fourier coefficient of this function on $a$ bits is $(1-a)2^{1-a}$ so the resulting contribution is $(1-a)2^{1-k}$.
\item With probability $2^{a-k}$, the president votes yes but all the remaining citizens vote no. If so, $f$ is $-1$ if and only if there is at most $1$ yes in the $a$ citizens. The Fourier coefficient of this function on $a$ bits is $(1-a)2^{1-a}$ so the resulting contribution is $(1-a)2^{1-k}$
\item With probability $(k-a-1)2^{a-k}$, the president and one citizen vote no but all the remaining citizens vote yes. If so, $f$ is $1$ if and only if all of the $a$ citizens vote yes. The Fourier coefficient of this function on $a$ bits is $2^{1-a}$ so the resulting contribution is $(k-a-1)2^{1-k}$
\item With probability $(k-a-1)2^{a-k}$, the president and one citizen vote yes but all the remaining citizens vote no. If so, $f$ is $-1$ if and only if all of the $a$ citizens vote no. The Fourier coefficient of this function on $a$ bits is $2^{1-a}$ so the resulting contribution is $(k-a-1)2^{1-k}$
\item In all other cases, the vote of these $a$ citizens does not matter.
\end{enumerate}
Summing these contributions up, we obtain $\frac{k - 2a}{2^{k-2}}$, as needed.
\end{proof}
\begin{lemma}
Whenever $a$ is even and $a \geq 2$, $\hat{f}_{P+aC} = \frac{2a-k}{2^{k-2}}$
\end{lemma}
\begin{proof}
To analyze this, consider the president and one citizen. There are several possibilities for what influence they can have:
\begin{enumerate}
\item The president gets his/her way. In this case, the contribution to the Fourier coefficient is $0$
\item Their votes don't matter because everyone else is unanimous. In this case, the contribution to the Fourier coefficient is also $0$.
\item $f$ will be $1$ if and only if at least one of the president and the citizen say yes. In this case, the contribution to the Fourier coefficient is $-\frac{1}{2}$
\item $f$ will only be $1$ if both the president and the citizen say yes. In this case, the contribution to the Fourier coefficient is $\frac{1}{2}$.
\end{enumerate} 
With this in mind, we have the following probabilities:
\begin{enumerate}
\item With probability $(a-1)2^{a+1-k}2^{1-a}$, all the outside citizens vote yes and $a-2$ of the remaining $a-1$ citizens vote yes. If so, $f$ is $1$ if and only if at least one of the president and the citizen say yes. The resulting contribution is $(a-1)2^{1-k}$.
\item With probability $(a-1)2^{a+1-k}2^{1-a}$, all the outside citizens vote no and $a-2$ of the remaining $a-1$ citizens vote no. If so, $f$ is $1$ if and only if at both the president and the citizen say yes. The resulting contribution is $(a-1)2^{1-k}$.
\item With probability $(k-a-1)2^{a+1-k}2^{1-a}$, all of the outside citizens except one vote yes and all of the remaining $a-1$ citizens vote yes. If so, $f$ is $1$ if and only if at least one of the president and the citizen say yes. The resulting contribution is $(1+a-k)2^{1-k}$.
\item With probability $(k-a-1)2^{a+1-k}2^{1-a}$, all of the outside citizens except one vote no and all of the remaining $a-1$ citizens vote no. If so, $f$ is $1$ if and only if both the president and the citizen say yes. The resulting contribution is $(1+a-k)2^{1-k}$.
\item In all other cases, the contribution to the Fourier coefficient is $0$.
\end{enumerate}
Summing these contributions up, we obtain $\frac{k - 2a}{2^{k-2}}$, as needed.
\end{proof}
\section{Inclusion/exclusion calculations}
\begin{lemma} \ 
\begin{enumerate}
\item $S_{\{i_1\},\{(i_2,i_3)\}} = {\beta}S_{\{(i_1,i_2)\}} - S_{\{i_1,(i_1,i_2)\}}$
\item $S_{\{\alpha\},\{(i_2,i_3)\}} = {\alpha}S_{\{(i_1,i_2)\}}$
\item $S_{\{i_1\},\{\alpha,i_2\}} = {\beta}S_{\{\alpha,i_1\}} - S_{\{i_1,(i_1,\alpha)\}}$
\item 
\begin{align*}
2S_{\{i_1\},\{(i_2,i_3)\},\{(i_4,i_5)\}} &= {\beta}(S_{\{(i_1,i_2)\}})^2 - 2S_{\{i_1,(i_1,i_2)\}}S_{\{(i_1,i_2)\}} - 2{\beta}S_{\{(i_1,i_2),(i_1,i_3)\}} -{\beta}S_{\{(i_1,i_2),(i_1,i_2)\}}\\
&+ 2S_{\{i_1,(i_1,i_2),(i_2,i_3)\}} + 4S_{\{i_1,(i_1,i_2),(i_1,i_3)\}} + 2S_{\{i_1,(i_1,i_2),(i_1,i_2)\}}
\end{align*}
\item 
\begin{align*}
2S_{\{\alpha\},\{(i_2,i_3)\},\{(i_4,i_5)\}} &= {\alpha}(S_{\{(i_1,i_2)\}})^2 - 2{\alpha}S_{\{(i_1,i_2),(i_1,i_3)\}} -{\alpha}S_{\{(i_1,i_2),(i_1,i_2)\}}
\end{align*}
\item 
\begin{align*}
S_{\{i_1\},\{(\alpha,i_2)\},\{(i_3,i_4)\}} &= {\beta}S_{\{(\alpha,i_1)\}}S_{\{(i_1,i_2)\}} - S_{\{i_1,(i_1,\alpha)\}}S_{\{(i_1,i_2)\}}- S_{\{i_1,(i_1,i_2)\}}S_{\{(\alpha,i_1)\}} \\
&- {\beta}S_{\{(\alpha,i_1),(i_1,i_2)\}} + S_{\{i_1,(i_1,i_2),(i_2,\alpha)\}} + 2S_{\{i_1,(i_1,\alpha),(i_1,i_2)\}}
\end{align*}
\item
\begin{align*}
6S_{\{i_1\},\{(i_2,i_3)\},\{(i_4,i_5)\},\{i_6,i_7\}} &= {\beta}(S_{\{(i_1,i_2)\}})^3 - 3S_{\{i_1,(i_1,i_2)\}}(S_{\{(i_1,i_2)\}})^2 - 6{\beta}S_{\{(i_1,i_2),(i_1,i_3)\}}S_{\{(i_1,i_2)\}} \\
&-3{\beta}S_{\{(i_1,i_2),(i_1,i_2)\}}S_{\{(i_1,i_2)\}} + 6S_{\{i_1,(i_1,i_2),(i_2,i_3)\}}S_{\{(i_1,i_2)\}} \\
&+ 12S_{\{i_1,(i_1,i_2),(i_1,i_3)\}}S_{\{(i_1,i_2)\}} + 6S_{\{i_1,(i_1,i_2),(i_1,i_2)\}}S_{\{(i_1,i_2)\}} \\
&+ 12{\beta}S_{\{(i_1,i_2),(i_1,i_3),(i_1,i_4)\}} + 12{\beta}S_{\{(i_1,i_2),(i_1,i_3),(i_2,i_3)\}} + 6{\beta}S_{\{(i_1,i_2),(i_2,i_3),(i_3,i_4)\}} \\
&+ 6{\beta}S_{\{(i_1,i_2),(i_1,i_2),(i_1,i_3)\}} + 2{\beta}S_{\{(i_1,i_2),(i_1,i_2),(i_1,i_2)\}} \\
&- 18S_{\{i_1,(i_1,i_2),(i_1,i_3),(i_1,i_4)\}} - 6S_{\{i_2,(i_1,i_2),(i_1,i_3),(i_1,i_4)\}} \\
&-12S_{\{i_1,(i_1,i_2),(i_1,i_3),(i_2,i_3)\}} - 6S_{\{i_2,(i_1,i_2),(i_2,i_3),(i_3,i_4)\}} \\
&- 12S_{\{i_1,(i_1,i_2),(i_1,i_2),(i_1,i_3)\}} - 3S_{\{i_2,(i_1,i_2),(i_1,i_2),(i_1,i_3)\}} \\
&- 3S_{\{i_3,(i_1,i_2),(i_1,i_2),(i_1,i_3)\}} - 3S_{\{i_1,(i_1,i_2),(i_1,i_2),(i_1,i_2)\}}
\end{align*}
\item
\begin{align*}
6S_{\{\alpha\},\{(i_2,i_3)\},\{(i_4,i_5)\},\{i_6,i_7\}} &= {\alpha}(S_{\{(i_1,i_2)\}})^3 - 6{\alpha}S_{\{(i_1,i_2),(i_1,i_3)\}}S_{\{(i_1,i_2)\}} -3{\alpha}S_{\{(i_1,i_2),(i_1,i_2)\}}S_{\{(i_1,i_2)\}} \\
&+ 12{\alpha}S_{\{(i_1,i_2),(i_1,i_3),(i_1,i_4)\}} + 12{\alpha}S_{\{(i_1,i_2),(i_1,i_3),(i_2,i_3)\}} + 6{\alpha}S_{\{(i_1,i_2),(i_2,i_3),(i_3,i_4)\}} \\
&+ 6{\alpha}S_{\{(i_1,i_2),(i_1,i_2),(i_1,i_3)\}} + 2{\alpha}S_{\{(i_1,i_2),(i_1,i_2),(i_1,i_2)\}}
\end{align*}
\item
\begin{align*}
2S_{\{i_1\},\{(\alpha,i_2)\},\{(i_3,i_4)\},\{i_5,i_6\}} &= {\beta}{S_{\{(\alpha,i_1)\}}}(S_{\{(i_1,i_2)\}})^2 - S_{\{i_1,(i_1,\alpha)\}}(S_{\{(i_1,i_2)\}})^2 - 2S_{\{i_1,(i_1,i_2)\}}{S_{\{(\alpha,i_1)\}}}S_{\{(i_1,i_2)\}} \\
&- 2{\beta}S_{\{(i_1,\alpha),(i_1,i_3)\}}S_{\{(i_1,i_2)\}} - 2{\beta}S_{\{(i_1,i_2),(i_1,i_3)\}}S_{\{(\alpha,i_1)\}} - {\beta}S_{\{(i_1,i_2),(i_1,i_2)\}}S_{\{(\alpha,i_1)\}} \\
&+ 2S_{\{i_1,(i_1,i_2),(i_2,\alpha)\}}S_{\{(i_1,i_2)\}} + 2S_{\{i_1,(i_1,i_2),(i_2,i_3)\}}S_{\{(\alpha,i_1)\}} \\
&+ 4S_{\{i_1,(i_1,\alpha),(i_1,i_2)\}}S_{\{(i_1,i_2)\}} + 4S_{\{i_1,(i_1,i_2),(i_1,i_3)\}}S_{\{(\alpha,i_1)\}} + 2S_{\{i_1,(i_1,i_2),(i_1,i_2)\}}S_{\{(\alpha,i_1)\}} \\
&+ 12{\beta}S_{\{(i_1,i_2),(i_1,i_3),(i_1,i_4)\}} + 12{\beta}S_{\{(i_1,i_2),(i_1,i_3),(i_2,i_3)\}} + 6{\beta}S_{\{(i_1,i_2),(i_2,i_3),(i_3,i_4)\}} \\
&+ 6{\beta}S_{\{(i_1,i_2),(i_1,i_2),(i_1,i_3)\}} + 2{\beta}S_{\{(i_1,i_2),(i_1,i_2),(i_1,i_2)\}} \\
&- 10S_{\{i_1,(i_1,\alpha),(i_1,i_2),(i_1,i_3)\}} - 2S_{\{i_2,(i_1,i_2),(i_1,\alpha),(i_1,i_3)\}} \\
&- 2S_{\{i_2,(\alpha,i_1),(i_1,i_2),(i_2,i_3)\}} - 2S_{\{i_3,(\alpha,i_1),(i_1,i_2),(i_2,i_3)\}}\\
&- 3S_{\{i_1,(i_1,i_2),(i_1,i_2),(i_1,\alpha)\}} - 2S_{\{i_2,(i_1,i_2),(i_1,i_2),(i_1,\alpha)\}}
\end{align*}
\end{enumerate}
\end{lemma}
\begin{proof}
\begin{align*}
2S_{\{i_1\},\{(i_2,i_3)\},\{(i_4,i_5)\}} &= {\beta}(S_{\{(i_1,i_2)\}})^2 - 2S_{\{i_1,(i_1,i_2)\}}S_{\{(i_1,i_2)\}} - 2{\beta}S_{\{(i_1,i_2),(i_1,i_3)\}} -{\beta}S_{\{(i_1,i_2),(i_1,i_2)\}}\\
&+ 2S_{\{i_1,(i_1,i_2),(i_2,i_3)\}} + 4S_{\{i_1,(i_1,i_2),(i_1,i_3)\}} + 2S_{\{i_1,(i_1,i_2),(i_1,i_2)\}}
\end{align*}
\begin{figure}[ht]
\centerline{\includegraphics[height=8cm]{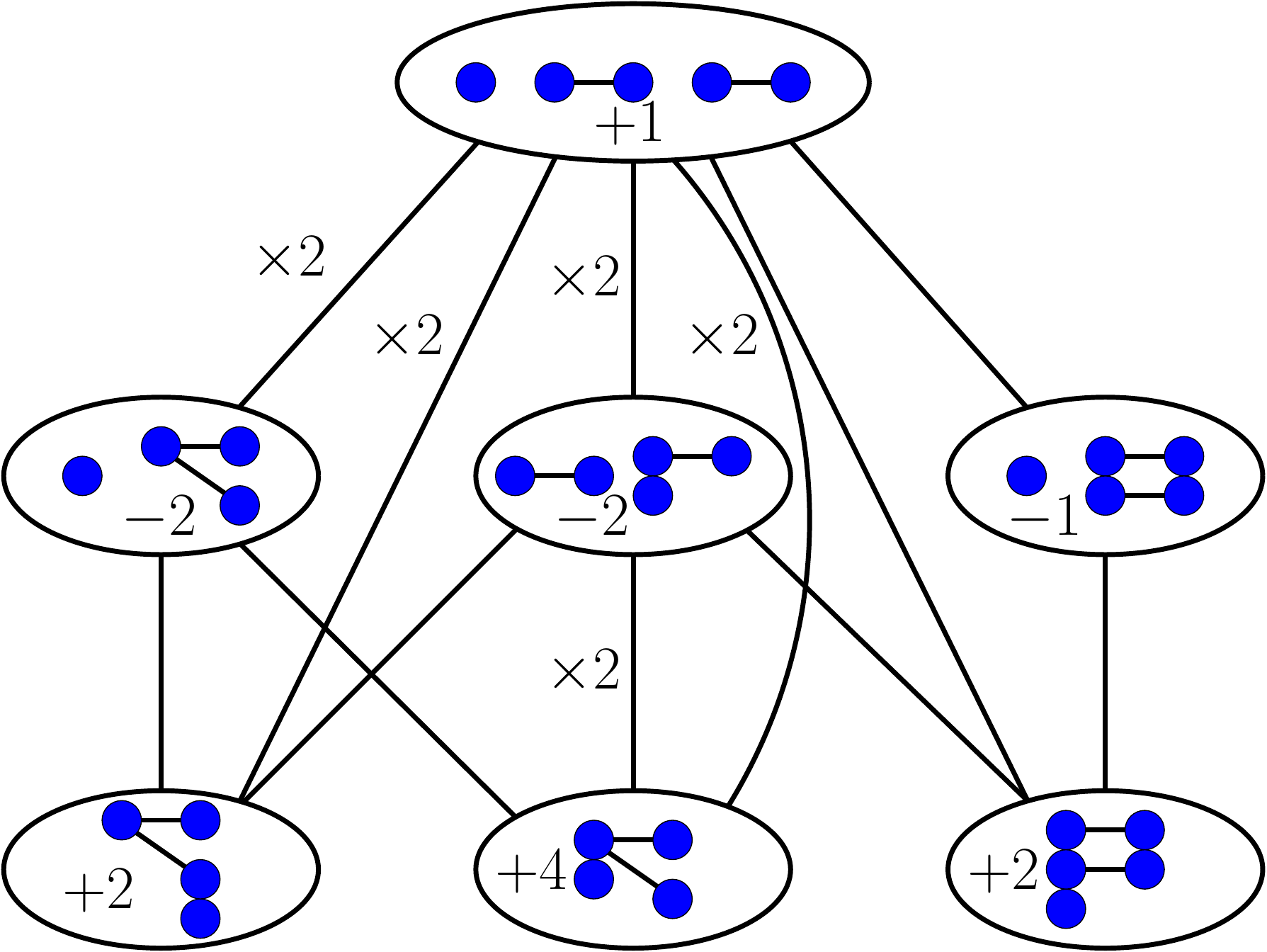}}
\caption{}
\label{iefigureone}
\end{figure}
\begin{align*}
2S_{\{\alpha\},\{(i_2,i_3)\},\{(i_4,i_5)\}} &= {\alpha}(S_{\{(i_1,i_2)\}})^2 - 2{\alpha}S_{\{(i_1,i_2),(i_1,i_3)\}} -{\alpha}S_{\{(i_1,i_2),(i_1,i_2)\}}
\end{align*}
\begin{figure}[ht]
\centerline{\includegraphics[height=5cm]{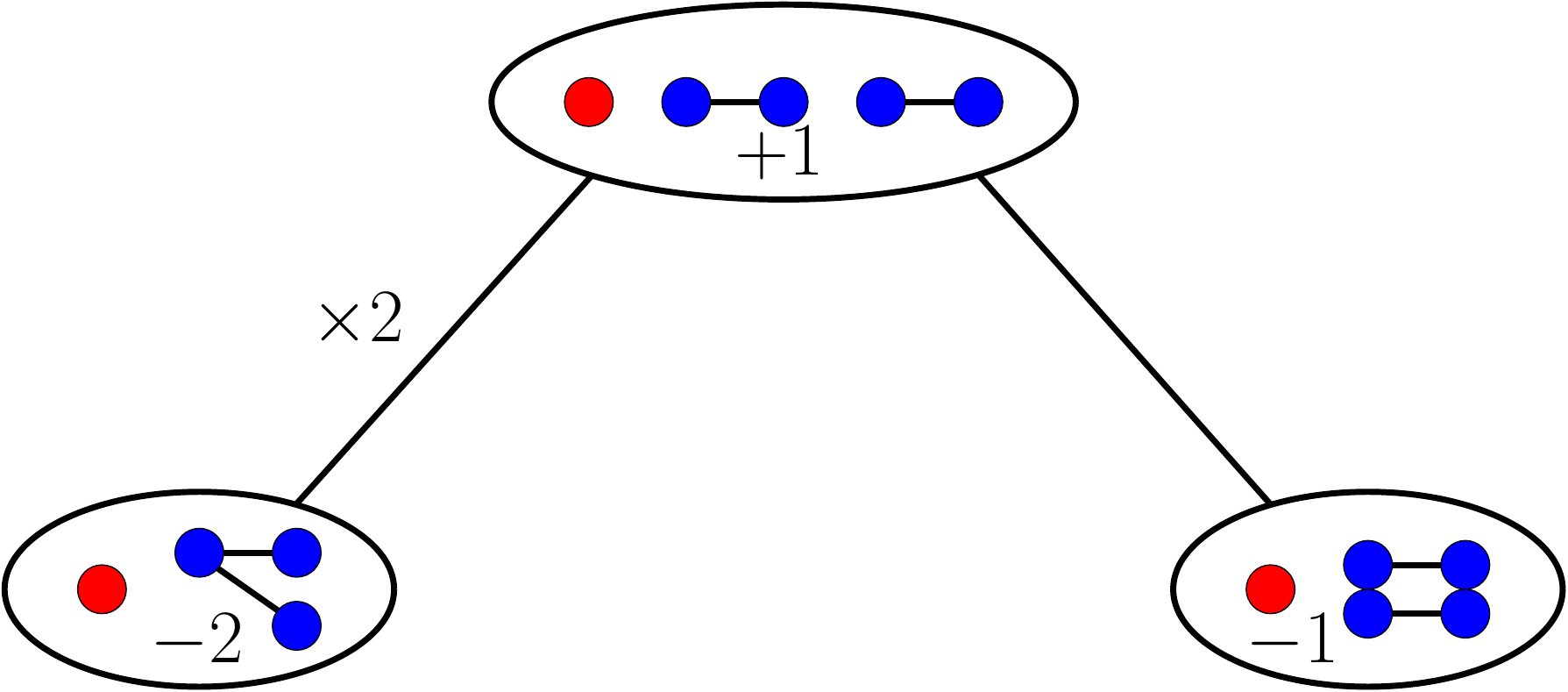}}
\caption{}
\label{iefiguretwo}
\end{figure}
\begin{align*}
S_{\{i_1\},\{(\alpha,i_2)\},\{(i_3,i_4)\}} &= {\beta}S_{\{(\alpha,i_1)\}}S_{\{(i_1,i_2)\}} - S_{\{i_1,(i_1,\alpha)\}}S_{\{(i_1,i_2)\}}- S_{\{i_1,(i_1,i_2)\}}S_{\{(\alpha,i_1)\}} \\
&- {\beta}S_{\{(\alpha,i_1),(i_1,i_2)\}} + S_{\{i_1,(i_1,i_2),(i_2,\alpha)\}} + 2S_{\{i_1,(i_1,\alpha),(i_1,i_2)\}}
\end{align*}
\begin{figure}[ht]
\centerline{\includegraphics[height=8cm]{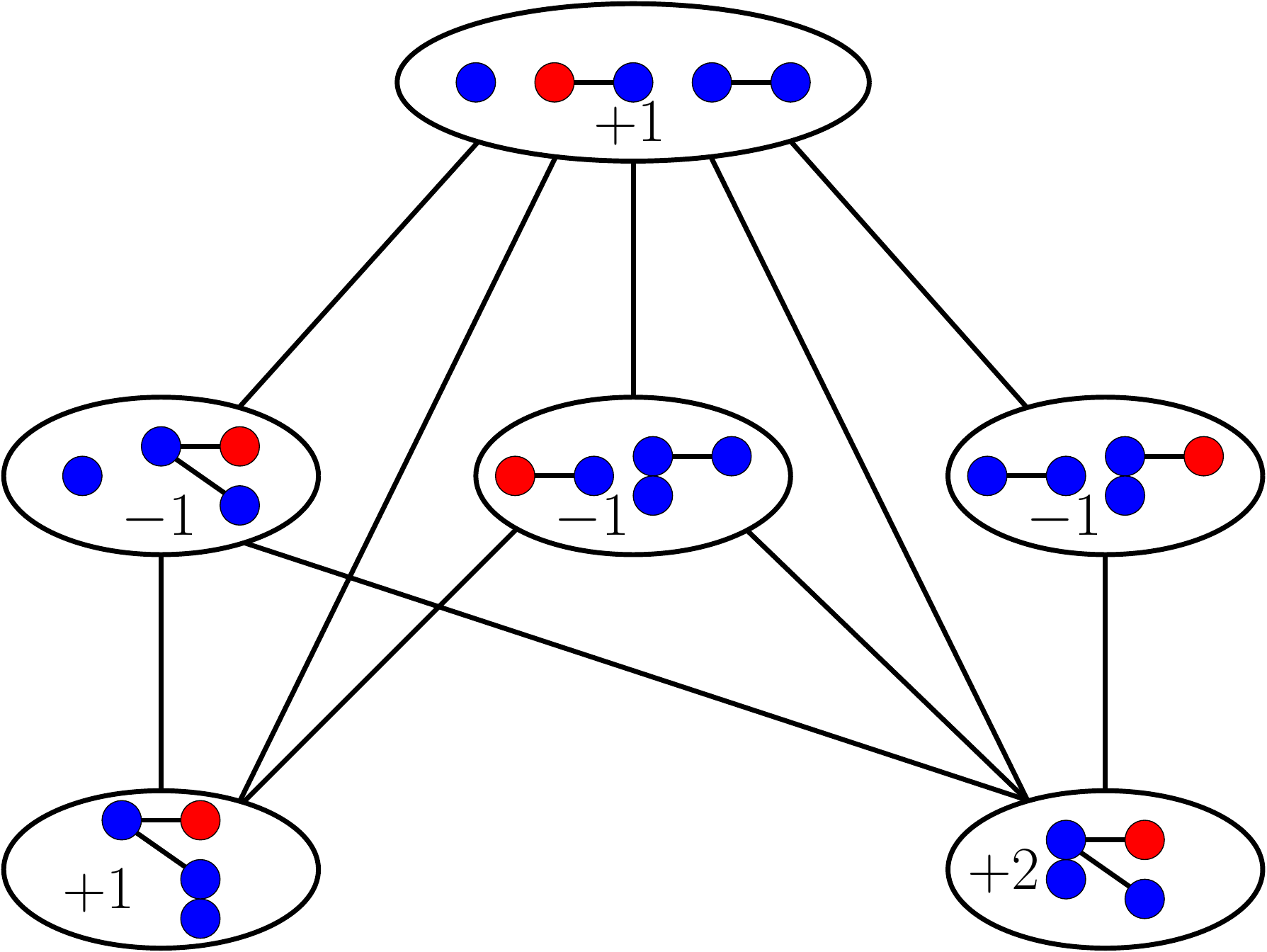}}
\caption{}
\label{iefigurethree}
\end{figure}
\begin{align*}
6S_{\{i_1\},\{(i_2,i_3)\},\{(i_4,i_5)\},\{i_6,i_7\}} &= {\beta}(S_{\{(i_1,i_2)\}})^3 - 3S_{\{i_1,(i_1,i_2)\}}(S_{\{(i_1,i_2)\}})^2 - 6{\beta}S_{\{(i_1,i_2),(i_1,i_3)\}}S_{\{(i_1,i_2)\}} \\
&-3{\beta}S_{\{(i_1,i_2),(i_1,i_2)\}}S_{\{(i_1,i_2)\}} + 6S_{\{i_1,(i_1,i_2),(i_2,i_3)\}}S_{\{(i_1,i_2)\}} \\
&+ 12S_{\{i_1,(i_1,i_2),(i_1,i_3)\}}S_{\{(i_1,i_2)\}} + 6S_{\{i_1,(i_1,i_2),(i_1,i_2)\}}S_{\{(i_1,i_2)\}} \\
&+ 12{\beta}S_{\{(i_1,i_2),(i_1,i_3),(i_1,i_4)\}} + 12{\beta}S_{\{(i_1,i_2),(i_1,i_3),(i_2,i_3)\}} + 6{\beta}S_{\{(i_1,i_2),(i_2,i_3),(i_3,i_4)\}} \\
&+ 6{\beta}S_{\{(i_1,i_2),(i_1,i_2),(i_1,i_3)\}} + 2{\beta}S_{\{(i_1,i_2),(i_1,i_2),(i_1,i_2)\}} \\
&- 18S_{\{i_1,(i_1,i_2),(i_1,i_3),(i_1,i_4)\}} - 6S_{\{i_2,(i_1,i_2),(i_1,i_3),(i_1,i_4)\}} \\
&-12S_{\{i_1,(i_1,i_2),(i_1,i_3),(i_2,i_3)\}} - 6S_{\{i_2,(i_1,i_2),(i_2,i_3),(i_3,i_4)\}} \\
&- 12S_{\{i_1,(i_1,i_2),(i_1,i_2),(i_1,i_3)\}} - 3S_{\{i_2,(i_1,i_2),(i_1,i_2),(i_1,i_3)\}} \\
&- 3S_{\{i_3,(i_1,i_2),(i_1,i_2),(i_1,i_3)\}} - 3S_{\{i_1,(i_1,i_2),(i_1,i_2),(i_1,i_2)\}}
\end{align*}
\begin{figure}[ht]
\centerline{\includegraphics[height=12cm]{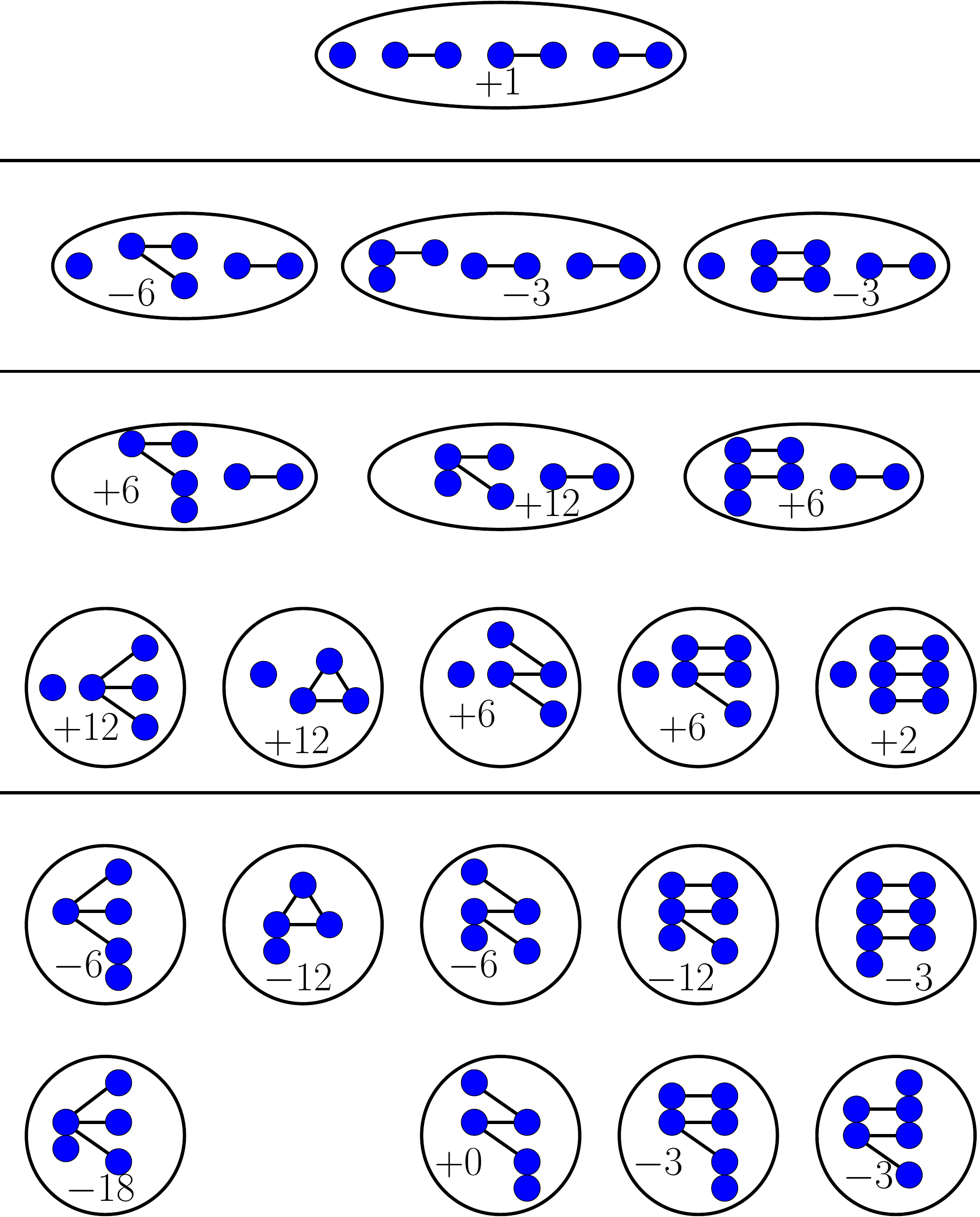}}
\caption{}
\label{iefigurefour}
\end{figure}
\begin{align*}
6S_{\{\alpha\},\{(i_2,i_3)\},\{(i_4,i_5)\},\{i_6,i_7\}} &= {\alpha}(S_{\{(i_1,i_2)\}})^3 - 6{\alpha}S_{\{(i_1,i_2),(i_1,i_3)\}}S_{\{(i_1,i_2)\}} -3{\alpha}S_{\{(i_1,i_2),(i_1,i_2)\}}S_{\{(i_1,i_2)\}} \\
&+ 12{\alpha}S_{\{(i_1,i_2),(i_1,i_3),(i_1,i_4)\}} + 12{\alpha}S_{\{(i_1,i_2),(i_1,i_3),(i_2,i_3)\}} + 6{\alpha}S_{\{(i_1,i_2),(i_2,i_3),(i_3,i_4)\}} \\
&+ 6{\alpha}S_{\{(i_1,i_2),(i_1,i_2),(i_1,i_3)\}} + 2{\alpha}S_{\{(i_1,i_2),(i_1,i_2),(i_1,i_2)\}}
\end{align*}
\begin{figure}[ht]
\centerline{\includegraphics[height=6cm]{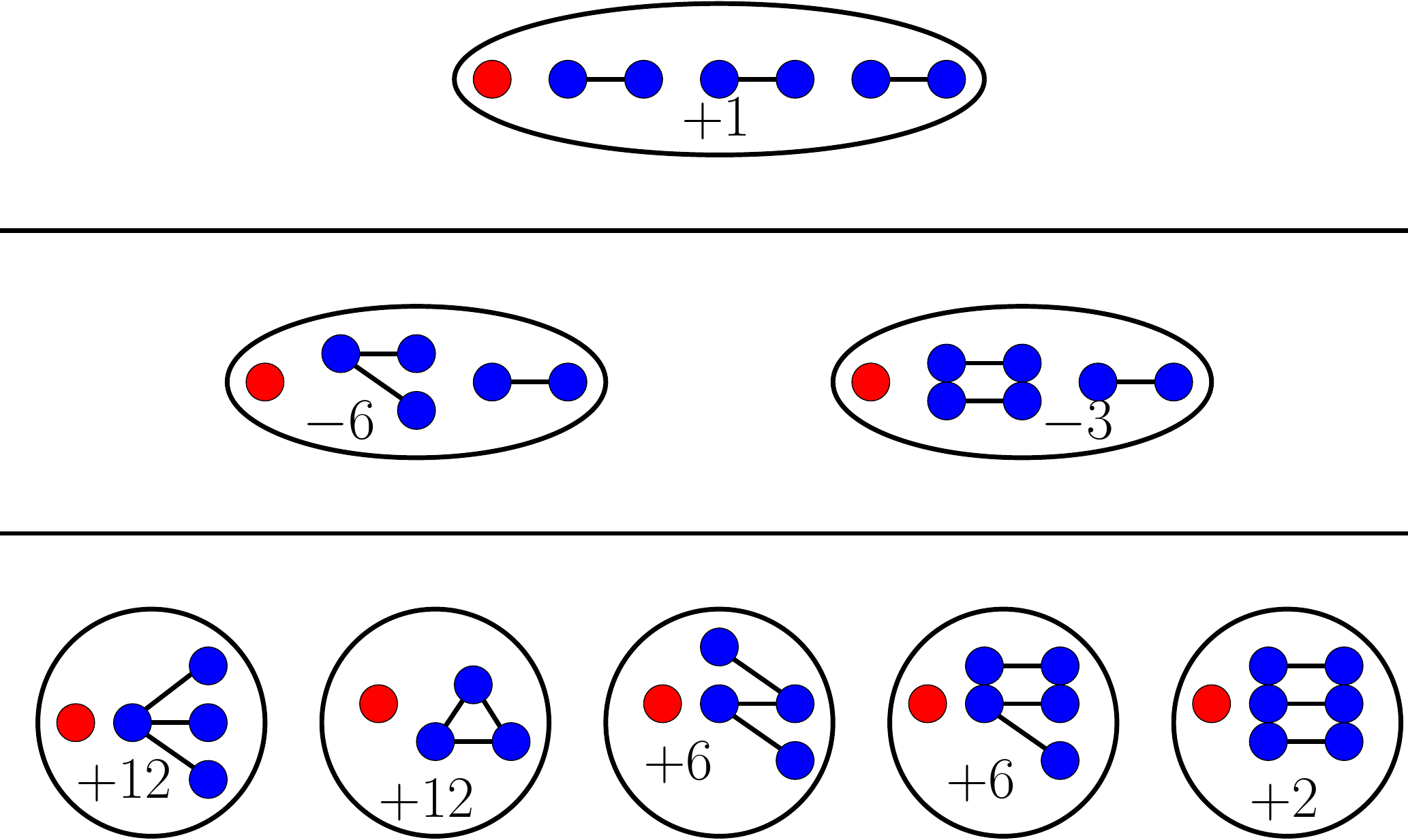}}
\caption{}
\label{iefigurefive}
\end{figure}
\begin{align*}
2S_{\{i_1\},\{(\alpha,i_2)\},\{(i_3,i_4)\},\{i_5,i_6\}} &= {\beta}{S_{\{(\alpha,i_1)\}}}(S_{\{(i_1,i_2)\}})^2 - S_{\{i_1,(i_1,\alpha)\}}(S_{\{(i_1,i_2)\}})^2 - 2S_{\{i_1,(i_1,i_2)\}}{S_{\{(\alpha,i_1)\}}}S_{\{(i_1,i_2)\}} \\
&- 2{\beta}S_{\{(i_1,\alpha),(i_1,i_3)\}}S_{\{(i_1,i_2)\}} - 2{\beta}S_{\{(i_1,i_2),(i_1,i_3)\}}S_{\{(\alpha,i_1)\}} - {\beta}S_{\{(i_1,i_2),(i_1,i_2)\}}S_{\{(\alpha,i_1)\}} \\
&+ 2S_{\{i_1,(i_1,i_2),(i_2,\alpha)\}}S_{\{(i_1,i_2)\}} + 2S_{\{i_1,(i_1,i_2),(i_2,i_3)\}}S_{\{(\alpha,i_1)\}} \\
&+ 4S_{\{i_1,(i_1,\alpha),(i_1,i_2)\}}S_{\{(i_1,i_2)\}} + 4S_{\{i_1,(i_1,i_2),(i_1,i_3)\}}S_{\{(\alpha,i_1)\}} + 2S_{\{i_1,(i_1,i_2),(i_1,i_2)\}}S_{\{(\alpha,i_1)\}} \\
&+ 12{\beta}S_{\{(i_1,i_2),(i_1,i_3),(i_1,i_4)\}} + 12{\beta}S_{\{(i_1,i_2),(i_1,i_3),(i_2,i_3)\}} + 6{\beta}S_{\{(i_1,i_2),(i_2,i_3),(i_3,i_4)\}} \\
&+ 6{\beta}S_{\{(i_1,i_2),(i_1,i_2),(i_1,i_3)\}} + 2{\beta}S_{\{(i_1,i_2),(i_1,i_2),(i_1,i_2)\}} \\
&- 10S_{\{i_1,(i_1,\alpha),(i_1,i_2),(i_1,i_3)\}} - 2S_{\{i_2,(i_1,i_2),(i_1,\alpha),(i_1,i_3)\}} \\
&- 2S_{\{i_2,(\alpha,i_1),(i_1,i_2),(i_2,i_3)\}} - 2S_{\{i_3,(\alpha,i_1),(i_1,i_2),(i_2,i_3)\}}\\
&- 3S_{\{i_1,(i_1,i_2),(i_1,i_2),(i_1,\alpha)\}} - 2S_{\{i_2,(i_1,i_2),(i_1,i_2),(i_1,\alpha)\}}
\end{align*}
\begin{figure}[ht]
\centerline{\includegraphics[height=12cm]{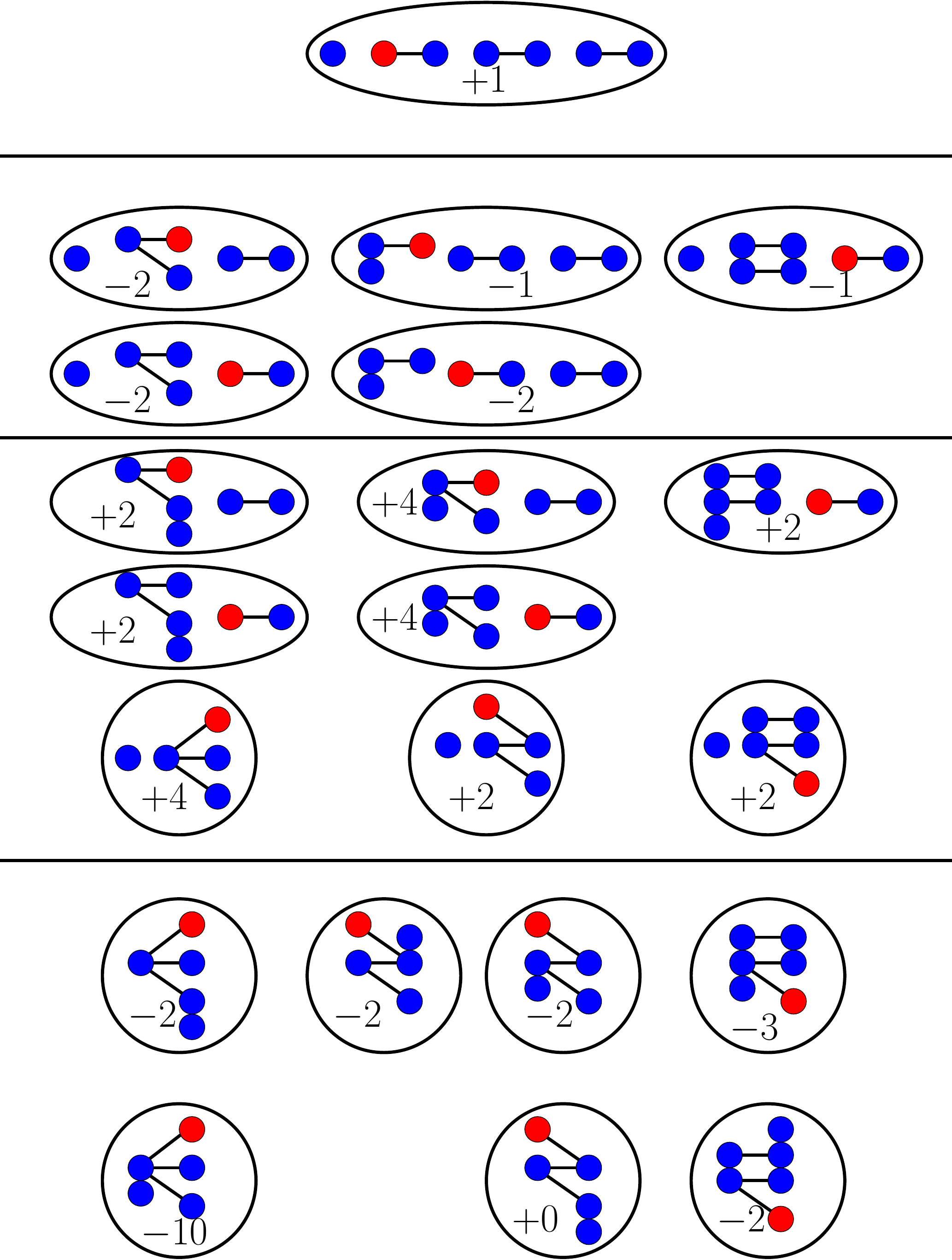}}
\caption{}
\label{iefiguresix}
\end{figure}
\end{proof}
\end{appendix}
\end{document}